\documentclass{article}
\PassOptionsToPackage{numbers}{natbib}

\usepackage[utf8]{inputenc} 
\usepackage[T1]{fontenc}    
\usepackage{hyperref}       
\usepackage{url}            
\usepackage{booktabs}       
\usepackage{amsfonts}       
\usepackage{nicefrac}       
\usepackage{microtype}      
\pdfoutput=1

\usepackage{natbib}
\usepackage{url}
\usepackage{amsmath}
 
\usepackage{amsthm}
\usepackage{algorithm}
\usepackage{algorithmicx}
\usepackage[noend]{algpseudocode}
\usepackage{verbatim}
\usepackage{graphicx}
\usepackage[space]{grffile}
\usepackage{tikz}
\usepackage{subfigure}
\usetikzlibrary{arrows}
\newtheorem{theorem}{Theorem}
\newtheorem*{theorem*}{Theorem}
\newtheorem{lemma}{Lemma}

\theoremstyle{definition}

\newtheorem{claim}{Claim}

\graphicspath{{figures/}}
\usepackage{amsmath}
\usepackage{mathrsfs}
\usepackage{amssymb}
\usepackage{booktabs}
\usepackage{changepage}
\usepackage{xspace}
\usepackage{hyperref}

\usepackage{tikz}
\usetikzlibrary{arrows}

\usepackage{xcolor}
\hypersetup{
    colorlinks,
    linkcolor={red!80!black},
    citecolor={blue!60!black},
    urlcolor={blue!80!black}
}

\newcommand{\nats}[0]{ \mathbb{N} }

\newcommand{\epsi}[0]{ \varepsilon }
\newcommand{\reals}[0]{ \mathbb{R}^+ }

\newcommand{\ff}[1]{ f \left( #1 \right) }

\newcommand{\fracc}[2]{ #1 / #2 }

\DeclareMathOperator*{\argmax}{arg\,max}

\newcommand{\opt}{\textsc{OPT}\xspace}

\newcommand{\oh}[1]{O\left( #1 \right)}
\newcommand{\ratio}{\left(\fracc{1}{(4c)} - \epsi \right)}

\newcommand{\ratioOne}{\left(\fracc{1}{4} - \epsi \right)}
\newcommand{\altratio}{\left( \frac{1}{1 + c + 1/(k^3 -1)} \right)\left( 1 - 1/e - (2c)/(ke) - c^2/(k^2e) \right)}

\newcommand{\algOne}{$\textsc{QuickStream}_c$\xspace}
\newcommand{\algOneone}{$\textsc{QuickStream}_1$\xspace}
\newcommand{\algTwo}{\textsc{BoostRatio}\xspace}

\newcommand{\uni}{\mathcal{N}}
\newcommand{\mon}{\textsc{SMCC}\xspace}
\newcommand{\ppass}{\textsc{P-Pass}\xspace}
\newcommand{\sstream}{\textsc{SieveStream++}\xspace}
\newcommand{\ck}{\textsc{C\&K}\xspace}
\newcommand{\greedy}{\textsc{Greedy}\xspace}
\newcommand{\lazy}{\textsc{LTL}\xspace}
\newcommand{\qs}{$\textsc{QuickStream}$\xspace}
\newcommand{\qss}{$\textsc{QuickSingleton}$\xspace}
\newcommand{\altqs}{$\textsc{QuickStreamLargeK}$\xspace}
\newcommand{\qsbr}{$\textsc{QS+BR}$\xspace}

\newcommand{\qsmem}{\oh{ ck \log(k) \log( 1 / \epsi )}}
\newcommand{\qsquery}{\lceil n / c \rceil + c}

\begin{document}
\title{Quick Streaming Algorithms for Maximization of Monotone Submodular Functions in Linear Time}
\author{Alan Kuhnle\thanks{Florida State University. Correspondence to \texttt{kuhnle@cs.fsu.edu}.}}
\date{September 10, 2020}
\maketitle
\begin{abstract}
We consider the problem of monotone, submodular maximization over a ground set of size $n$ subject to cardinality constraint $k$. For this problem, we introduce the first deterministic algorithms with linear time complexity; these algorithms are streaming algorithms. Our single-pass algorithm obtains a constant ratio in $\lceil n / c \rceil + c$, for any $c \ge 1$. In addition, we propose a deterministic, multi-pass streaming algorithm with a constant number of passes that achieves nearly the optimal ratio with linear query and time complexities. We prove a lower bound that implies no constant-factor approximation exists using $o(n)$ queries, even if queries to infeasible sets are allowed. An empirical analysis demonstrates that our algorithms require fewer queries (often substantially less than $n$) yet still achieve better objective value than the current state-of-the-art algorithms, including single-pass, multi-pass, and non-streaming algorithms.
\end{abstract}
\section{Introduction} \label{sec:intro}
\begin{table*}[t] \caption{State-of-the-art algorithms for \mon in terms of time complexity.  } \label{table:cmp}
\begin{center} \scriptsize
\begin{tabular}{llllll} \toprule 
Reference         & Passes & Ratio & Memory & Queries & Time \\
  \midrule
  \lazy \citep{Mirzasoleiman2014} & $k$ & $1 - 1/e - \epsi$ & $O(n)$ & $n \log( 1 / \epsi)$ & $O(n)$ \\
  \ppass \citep{Jakub2018} & $O( 1 / \epsi )$ & $1 - 1/e - \epsi$ & $O( k \log (k) / \epsi )$ & $O( n \log (k) / \epsi^2 ) $ & $O(n \log k)$ \\ 
  \sstream \citep{Kazemi2019} & 1 & $1/2 - \epsi$ & $O( k / \epsi )$ & $O(n \log(k) / \epsi )$ & $O(n \log k)$ \\
  \ck \citep{Chakrabarti2015} & 1 & 1/4 & $O(k)$ & $2n$ & $O(n \log k)$  \\ \midrule
  $\qs_c$, $c \ge 1$ (Theorem \ref{thm:online})        & 1 &$1/(4c) - \epsi$ & $\qsmem$ & $\lceil n / c \rceil + c$ & $O(n)$ \\
  \qsbr (Theorem \ref{thm:optimal-det})          & $O(1/ \epsi)$ &$1 - 1/e - \epsi$ & $O(k\log (k))$ & ${O(n/{\epsi})}$ & $O(n)$ \\ \bottomrule
\end{tabular}
\end{center}
\end{table*}k
A nonnegative, set function $f:2^{\mathcal U} \to \reals$, where ground set $\mathcal U$ is of 
size $n$, is \textit{submodular}
if for all $S \subseteq T \subseteq \mathcal U$, $u \in \mathcal U \setminus T$, 
$\ff{ T \cup \{ u \} } - \ff{ T } \le \ff{ S \cup \{u \} } - \ff{S}$
and
\textit{monotone} if $f(A) \le f(B)$ if $A \subseteq B$.
Intuitively,
submodularity captures a natural diminishing returns property that arises in many
machine learning applications, such as
viral marketing \citep{Kempe2003}, 
network monitoring \citep{Leskovec2007},
sensor placement \citep{Krause2007},
video summarization \citep{Mirzasoleiman2018}, 
and MAP Inference
for Determinantal Point Processes \citep{Gillenwater2012}.

A well-studied NP-hard optimization problem
in this context is submodular maximization subject to a cardinality constraint (\mon):
$\argmax_{ |S| \le k } f(S),$
where the cardinality
constraint $k$ is an input parameter and the function $f$ is submodular and monotone.
A simple greedy procedure \citep{Nemhauser1978} achieves approximation
ratio of $1 - 1/e \approx 0.632$ for SMCC in $O( kn )$ time; this ratio
is optimal under the value query model \citep{Nemhauser1978a}.
In the value query model, the function $f$ is provided to an algorithm as a value oracle,
which when queried with set $S$ returns $f(S)$ in a single
operation that requires $O(1)$ time. In this work, the time
complexity of an algorithm is measured in terms of the number
of arithmetic operations and number of oracle queries.

For $k = \Omega (n)$,
the standard greedy algorithm has $\Omega( n^2 )$ time complexity,
which is prohibitive on modern instance sizes. Further, loading
the entire ground set into memory may be impossible. Therefore,
much effort has gone into the design of algorithms with lower time
complexity \citep{Badanidiyuru2014,Mirzasoleiman2014,Buchbinder2015a,
Kuhnle2019,Crawford2020}; and into streaming algorithms \citep{Gomes2010,Badanidiyuru2014a,Chakrabarti2015}.
In this context, a \emph{streaming algorithm}\footnote{Formally, this is the semi-streaming model since $k$ could be large relative to $n$. In this work, we will assume each element of the ground set requires $O(1)$ space.} accesses elements by one or more sequential passes 
through the ground set and stores at most $O(k \log(n))$ elements in memory.

Several randomized approximation algorithms
\citep{Mirzasoleiman2014,Buchbinder2015a,Fahrbach2018}
have been designed that require $O(n)$ time, independent of $k$.
However, the ratios of these algorithms hold only in expectation, 
which is undesirable for applications in which a good solution is required with high probability.
Furthermore, these algorithms are not streaming algorithms and
require the entire ground set to be loaded into memory. 
Indeed,
every deterministic or streaming algorithm with constant ratio  
that has been described in the literature
requires $\Omega( n \log k )$ time.
This statement remains true if ``deterministic'' is replaced
by ``with high probability'' (that is, probability that converges
to $1$ as $n \to \infty$).
Moreover, every deterministic or streaming algorithm
requires $\Omega(n \log k)$ queries to the value oracle,
except for the single-pass streaming algorithm 
of \citet{Chakrabarti2015}, which obtains a ratio of $1/4$ 
with $2n$ oracle queries and $O(n \log k)$ arithmetic
operations.




\paragraph{Contributions}
In this work, we propose
    the first deterministic, streaming algorithms for \mon that have linear time complexity 
    in the size $n$ of the ground set. The first algorithm is a single-pass streaming algorithm that obtains
    a constant ratio, and the second is a multi-pass streaming algorithm that obtains nearly the optimal ratio. Specifically:
\begin{itemize}
  \item We provide a linear-time, single-pass algorithm \qs (Section \ref{sec:qs} and Appendix \ref{apx:qs}), which achieves a constant ratio of while
    making at most $\lceil n/c \rceil + c$ queries to the value oracle for $f$, for any $c \ge 1$.
    This is the lowest query complexity\footnote{The query complexity of an algorithm is the total number of queries made to the value oracle for $f$ and is upper-bounded by the time complexity.} of any constant factor algorithm, which is
    important as the cost to evaluate the function $f$ may be expensive. 
    The following theorem summarizes the guarantees for \qs.
    \begin{theorem} \label{thm:online}
  Let $c \ge 1$ be an integer, and let $\epsi > 0$.
  There exists a deterministic, single-pass streaming algorithm that makes at most
  $\qsquery$ queries,
  has memory complexity $\qsmem$ has 
  approximation ratio at least $1 / (4c) - \epsi$ for \mon,
  and the ratio converges to $(1 - 1/e)/(c + 1)$ as $k \to \infty$.
  Further, the time complexity of the algorithm is $O(n)$.
\end{theorem}
We also show a lower bound of $\Omega( n / k )$ 
on the time complexity to obtain a constant ratio (Section \ref{apx:thm2}). 
  \item We propose a multi-pass algorithm \qsbr (Section \ref{sec:brlarge}), which achieves nearly the optimal ratio $1 - 1/e - \epsi$ 
    in a constant number of passes and linear time complexity. In addition, this 
    algorithm is the first deterministic algorithm for SMCC to
    obtain nearly the optimal ratio with a linear query complexity. 
    \begin{theorem} \label{thm:optimal-det}
  There exists a deterministic, multi-pass streaming algorithm for
  \mon that achieves approximation ratio $1 - 1/e - \epsi$,
  makes $O( n / \epsi )$ oracle queries, requires
  $O( 1 / \epsi )$ passes over the ground set, and requires
  $O( k \log k )$ memory. Further, the time complexity of the
  algorithm is $O(n)$.
\end{theorem}
  \item An empirical evaluation (Section \ref{sec:exp}) of our single-pass algorithm $\qs$
shows that if $\qs$ is supplemented with a linear-time post-processing procedure
(which does not compromise any of the theoretical guarantees of the algorithm),
it empirically 
exceeds the objective value of the state-of-the-art single-pass streaming algorithm
\sstream \citep{Kazemi2019} and the non-streaming \lazy algorithm, while using
fewer queries than either algorithm. Further, \qsbr obtains an even
greater objective value while remaining query efficient.
\end{itemize}
Table \ref{table:cmp} shows how our algorithms compare theoretically
to the current state-of-the-art algorithms for \mon. The source code used in the empirical evaluation is available at: \url{https://gitlab.com/kuhnle/linear-submodular-stream}.

\subsection{Related Work} \label{sec:rw}
The literature studying \mon is vast, so we only discuss algorithms for
\mon with monotone objective and cardinality constraint in this section. 
Streaming algorithms for more generalized constraints and submodular
but not necessarily monotone functions include the works of
\citet{Chekuri2015}, \citet{Mirzasoleiman2016}, \citet{Mirzasoleiman2018},
and \citet{Feldman2018}, among others.
\paragraph{Fast Approximation Algorithms} 
The stochastic greedy algorithm \lazy of \citet{Mirzasoleiman2014} obtains
a ratio of $1 - 1/e - \epsi$ in $O(n)$ time, and thus has nearly optimal ratio
and time complexity. However, 
its ratio holds only in expectation: \lazy returns a poor solution with constant probability if $k = O(1)$. We refer the reader to \citet{Hassidim2017} for discussion and further analysis of the ratio of \lazy; also, in Section \ref{sec:exp},
we empirically explore the behavior of \lazy for large values of $\epsi$.
In addition to \lazy, two other randomized approximation algorithms
with linear query and time complexities have been developed. 
The algorithm 
of \citet{Buchbinder2015a} achieves ratio $1/e - \epsi$
in $O(n \log (1 / \epsi) / \epsi^2)$ time.
Very recently, the randomized, parallelizable algorithm of \citet{Fahrbach2018}
obtains ratio $1 - 1/e - \epsi$ in expectation with
time complexity $O(n \log( 1 / \epsi ) / \epsi^3)$. 
In contrast to our algorithms,
none of these algorithms are streaming algorithms
or are deterministic. For some applications of SMCC, an approximation 
ratio that holds only in expectation (rather than deterministically or
with high probability) may be undesirable.

\paragraph{Single-Pass Streaming Algorithms}
\citet{Chakrabarti2015} provided the first single-pass streaming algorithm for
\mon; they designed a $(1/4)$-approximation with one pass, $2n$ total 
queries, and $O(k)$ memory. However, this algorithm
requires time complexity of $\Omega(n \log k)$.
\citet{Badanidiyuru2014a} improved the ratio for a single-pass algorithm
to $1/2 - \epsi$ in
$O( k \log(k) / \epsi )$ memory, and $O( n\log(k) / \epsi )$ total queries and time.
\citet{Kazemi2019} have provided the single
pass $1/2 - \epsi$ approximation
\sstream, which improves the
algorithm of \citet{Badanidiyuru2014a} to have
memory complexity of $O(k / \epsi)$ as indicated in Table \ref{table:cmp}. 
The current state-of-the-art, single-pass algorithm is
\sstream, which is empirically compared to our algorithms
in Section \ref{sec:exp}.
Finally,
\citet{Feldman2020} recently showed that any one-pass algorithm
with approximation guarantee of $1/2 + \epsi$ must essentially store
all elements of the stream. 
In contrast to our single-pass algorithm, none of these algorithms have linear
time complexity. Further, they require more oracle queries by at least a constant
factor.

\paragraph{Multi-Pass Streaming Algorithms}
The first multi-pass streaming algorithm for \mon 
has been given by \citet{Gomes2010}, which obtains
value $\opt / 2 - k \epsi$ using $O(k)$ memory and 
$O( B / \epsi )$ passes, where $f$ is upper bounded
by $B$. 
\citet{Jakub2018} designed a multi-pass algorithm \ppass that obtains
ratio $1 - 1/e - \epsi$ in $O( 1 / \epsi )$ passes, 
$O( k \log (k) / \epsi )$ memory,
$O(n \log (k) / \epsi^2 ) $ time. This is a
generalization of the multi-pass algorithm
of \citet{McGregor2019} for the maximum coverage
problem. 
The current state-of-the-art, multi-pass algorithm is
\ppass, which is empirically compared to our algorithms
in Section \ref{sec:exp}. In contrast to our multi-pass algorithm,
no multi-pass algorithm has linear time complexity; further,
our algorithm makes fewer passes than \ppass to achieve the same ratio
of $1 - 1/e - \epsi$.

\section{The $\qs_c$ Algorithm} \label{sec:qs}
The algorithm \algOne 
is a single-pass, deterministic streaming algorithm.
The parameter $c$ is the number of elements buffered before the
algorithm processes them together; this parameter determines the approximation
ratio, query complexity, and memory complexity of the algorithm:
respectively, $1/(4c)$, $\lceil n / c \rceil + c$, and $\qsmem$.
Notably, this algorithm is the first deterministic algorithm for SMCC
to obtain linear time complexity.
To handle the case that $k = 1$ and
obtain better ratios if $k \ge 8c / e$, we provide two
related algorithms in Appendix \ref{apx:qs}.



\begin{algorithm}[t]
   \caption{For each $c \ge 1$, a single-pass algorithm with approximation ratio $\ratio$ if $k \ge 2$, query complexity $\lceil n / c \rceil + c$, and memory complexity $\qsmem$.} \label{alg:quickstream}
   \begin{algorithmic}[1]
     \Procedure{\algOne}{$f, k, \epsi$}
     \State \textbf{Input:} oracle $f$, cardinality constraint $k$, $\epsi > 0$
     \State $A \gets \emptyset$, $A' \gets \emptyset$, $C \gets \emptyset$, $\ell \gets \lceil \log_2 ( 1 / (4\epsi) ) \rceil  + 3$
     \For{ element $e$ received }
     \State $C \gets C \cup \{ e \}$
     \If{$|C| = c$ or stream has ended}
     \If{$f( A \cup C ) - f(A) \ge f(A) / k$}\label{line:add}
     \State $A \gets A \cup C$
     \EndIf
     \If{$|A| > 2c \ell (k + 1) \log_2 (k)$}
     \State $A \gets \{ c\ell (k + 1)\log_2 (k)$ elements most recently added to $A \}$\label{line:delete-A}
     \EndIf
     \State $C \gets \emptyset$
     \EndIf
     \EndFor
     \State $A' \gets \{ ck \text{ elements most recently added to } A \}$. \label{line:secondtolast}
     \State Partition $A'$ arbitrarily into at most $c$ sets of size at most $k$. Return the set of the partition with highest $f$ value.\label{line:term}
     \EndProcedure
\end{algorithmic}
\end{algorithm}

The algorithm $\qs_c$ maintains a set $A$, initially empty. 
We refer to the sets of size at most $c$ of elements
processed together as \emph{blocks} of size $c$. When
a new block $C$ is received, the algorithm makes one query
of $f(A \cup C)$. If $f(A \cup C) - f(A) \ge f(A) / k$,
the block $C$ is added to $A$; otherwise, it is discarded.
If the size $|A|$ exceeds $2\ell c(k + 1) \log_2 k$, elements are
deleted from $A$. 
At the end of the stream,
the algorithm partitions the last $ck$ elements added to $A$
into $c$ pieces of size at most $k$ and return the one
with highest $f$ value. 
Pseudocode is given
in Alg. \ref{alg:quickstream}. 

At a high level, our algorithm
resembles a swapping algorithm
such as \citet{Chakrabarti2015} or
\citet{Buchbinder2015b}, which
replaces previously added elements with
better ones as they arrive. However, our
algorithm uses simply the order in which
elements were added to $A$ to compare elements;
which bypasses the need of a direct comparison of the value of
an incoming element with the other elements of $A$.
This indirect method of comparison allows us to obtain
an algorithm with linear time complexity. 

Below, we prove the following theorem.
\begin{theorem} \label{thm:qs}
  Let $c \ge 1$, $\epsi \ge 0$, and let $(f,k)$ be an instance of \mon
  with $k \ge 2$.
  The solution $S$ returned by $\qs_c$ satisifes 
  $f(S) \ge \ratio \opt,$
  where $\opt$ is the optimal solution value on this instance.
  Further, $\qs_c$ makes at most $\lceil n / c \rceil + c$ queries
  and has memory complexity $\qsmem$.
\end{theorem}

We remark that using the the value $f(A)$ of a potentially infeasible set $A$ 
is an important feature of our algorithm; 
the use of infeasible sets is necessary to obtain a constant ratio
with fewer than $n$ queries to the oracle. 

\begin{proof}[Proof of Theorem \ref{thm:qs}]
The query complexity, time complexity, and memory complexity of $\qs_c$ are clear
from the limit on the size of $A$, the choice of $\ell$, and the
fact that one query is required every $c$ elements together with
$c$ queries at the termination of the stream. The rest of the proof
establishes the approximation ratio of $\qs_c$.

First, we argue it is sufficient to prove the ratio in the case $c = 1$.
Let $\uni = \{C_1, \ldots, C_m \}$,
where each $C_i$ is the $i$-th block of at most $c$ elements of $\mathcal U$
considered for addition to $A$ on line \ref{line:add}.
Define monotone, submodular function $g:2^\uni \to \reals$ by 
$g( S ) = f( \bigcup_{ C \in S } C )$.
Observe that if we omit lines \ref{line:secondtolast} and \ref{line:term}, 
the behavior of \algOne on instance $(f,k)$ is equivalent to \algOneone
run on instance $(g, k)$ of \mon; further, $\argmax_{|S| \le k} g(S) \ge \argmax_{|S| \le k} f(S)$.
Let $S$ be the solution returned by \algOneone on instance $(g,k)$. 
Then the value of $A'$ at termination of \algOne is $A' = \bigcup_{C \in S} C$. Let $\{D_1, \ldots, D_c\}$ be the partition of $A$ on line \ref{line:term} of Alg. \ref{alg:quickstream}.
Then by submodularity of $f$
$$g( S ) = f( A' ) \le \sum_{i = 1}^c f(D_i) \le c \max_{1 \le i \le c} f(D_i).$$
Since \algOne returns $\argmax_{1 \le i \le c} f(D_i)$,
it suffices to show that \algOneone has approximation ratio $\ratioOne$.

For the rest
of the proof, we let $c = 1$. 
We require the following claim, which follows directly
from the inequality $\log x \ge 1 - 1/ x$ for $x > 0$.
\begin{claim} \label{claim:elementary}
For $y \ge 1$, if $i \ge (k + 1) \log y$, then $(1 + 1/k)^i \ge y$.
\end{claim}
Throughout the proof,
let $A_i$ denote the value of $A$ at the beginning
of the $i$-th
iteration of \textbf{for} loop; let $A_{n + 1}$ be the value of $A$
after the \textbf{for} loop completes. Also, let $A^* = \bigcup_{1 \le i \le n + 1} A_i$, and
let $e_i$ denote the element received at the beginning of iteration $i$. We refer
to line numbers of the pseudocode Alg. \ref{alg:quickstream}.
First, we show the value of $f(A)$ does not decrease
between iterations of the \textbf{for} loop, despite the possibility of 
deletions from $A$. 
\begin{lemma} \label{lemm:increasing}
  For any $1 \le i \le n$, it holds that $f(A_{i}) \le f(A_{i + 1})$.
\end{lemma}
\begin{proof}
  If no deletion is made during iteration $i$
  of the \textbf{for} loop, then any change
  in $f(A)$ is clearly nonnegative. So suppose 
  deletion of set $B$ from $A$ occurs on line \ref{line:delete-A} of Alg. \ref{alg:quickstream} during
  this iteration. Observe that $A_{i + 1} = (A_i \setminus B) \cup \{ e_i \}$,
  because the deletion is triggered by the addition of $e_i$ to $A_i$.
  In addition, at some iteration 
  $j < i$ of the \textbf{for} loop, it holds that $A_j = B$. 
  From the beginning of iteration $j$ to the beginning of iteration $i$,
  there have been $\ell(k + 1) \log_2 (k) - 1 \ge (\ell - 1)(k + 1) \log_2 (k)$ additions and no deletions to $A$, which add
  precisely the elements 
  in $(A_i \setminus A_j)$.

  It holds that
  \begin{align*}
    \ff{ A_i \setminus A_j } &\overset{(a)}{\ge} \ff{ A_i } - \ff{ A_j } \\ 
    &\overset{(b)}{\ge} \left( 1 + \frac{1}{k} \right)^{(\ell - 1)(k + 1) \log k} \cdot f( A_j ) - f (A_j) \\ &\overset{(c)}{\ge} (k^{\ell - 1} - 1 ) f( A_j ),
  \end{align*}
  where inequality (a) follows from submodularity and nonnegativity of $f$,
  inequality (b) follows from the fact that each addition from $A_j$ to $A_i$ increases
  the value of $f(A)$ by a factor of at least $(1 + 1/k)$, and inequality (c)
  follows from Claim
  \ref{claim:elementary}.
  Therefore  %
  \begin{align}
    f(A_i) &\le \ff{A_i \setminus A_j} + \ff{ A_j } \nonumber \\ &\le \left( 1 + \frac{1}{k^{\ell - 1} - 1} \right) \ff{A_i \setminus A_j}. \label{ineq:1}
  \end{align}
  Next,
  \begin{align} \label{ineq:2}
    &\ff{ (A_i \setminus A_j) \cup \{ e_i \} } - \ff{ A_i \setminus A_j } \nonumber \\ 
    &\overset{(d)}{\ge} \ff{ A_i \cup \{ e_i \} } - \ff{ A_i } \nonumber \\ 
    &\overset{(e)}{\ge} \ff{A_i} / k  \ge \ff{A_i \setminus A_j } / k, \end{align}
  where inequality (d) follows from submodularity, and inequality (e) is by
  the condition to add $e_i$ to $A_i$ on line \ref{line:add}.
  Finally, using Inequalities (\ref{ineq:1}) and (\ref{ineq:2}) as indicated below, we have
  \begin{align*}
    \ff{ A_{i + 1}} &= \ff{ A_i \setminus A_j \cup \{ e_i \} } \\ 
    &\overset{\text{By (\ref{ineq:2})}}{\ge} \left( 1 + \frac{1}{k} \right) \ff{ A_i \setminus A_j } \\ 
    &\overset{\text{By (\ref{ineq:1})}}{\ge} \frac{1 + \frac{1}{k}}{1 + \frac{1}{k^{\ell - 1} - 1}} \cdot f(A_i) \ge f( A_i ), 
  \end{align*}
  where the last inequality follows since $k \ge 2$ and $\ell \ge 3$.
\end{proof}
Next, we bound the total value of $f( A )$ lost from
deletion throughout the run of the algorithm.
\begin{lemma} \label{lemm:aplus}
  $\ff{A^*} \le \left( 1 + \frac{1}{k^\ell - 1} \right)\ff{ A_{n + 1} }. $
\end{lemma}
\begin{proof}
  Observe that $A^* \setminus A_{n + 1}$ may be written as the
  union of pairwise disjoint sets, each of which is size
  $\ell(k+1) \log_2(k) + 1$ and was deleted on line \ref{line:delete-A} of Alg. \ref{alg:quickstream}.
  Suppose there were $m$ sets deleted from $A$; write 
  $A^* \setminus A_{n + 1} = \{ B^i : 1 \le i \le m \}$, where
  each $B^i$ is deleted on line \ref{line:delete-A}, ordered such that $i < j$
  implies $B^i$ was deleted after $B^j$ (the reverse order in which they were deleted); finally, let $B^0 = A_{n + 1}$.
  \begin{claim} \label{claim:deletion}
    Let $0 \le i \le m$. Then $\ff{B^i} \ge k^{\ell} \ff{B^{i+1}}$.
  \end{claim}
  \begin{proof}
    Let $B^i$, $B^{i + 1} \in \mathcal B$. 
    There are at least $\ell(k + 1) \log k + 1$ elements added to $A$ and exactly one deletion
    event during the period between starting when $A = B^{i + 1}$ until $A = B^i$. 
    Moreover, each addition except possibly one (corresponding to the deletion event)
    increases $f(A)$ by a factor
    of at least $1 + 1/k$. Hence, by Lemma \ref{lemm:increasing} and Claim \ref{claim:elementary},
  $\ff{B^i} \ge k^{\ell} \ff{B^{i+1}}$.
  \end{proof}
  By Claim \ref{claim:deletion}, for any $0 \le i \le m$,
  $\ff{ A_{n + 1} } \ge k^{\ell i} \ff{ B^i }$.
  Thus, 
  \begin{align*}
    \ff{ A^* } &\le \ff{ A^* \setminus A_{n + 1} } + \ff{ A_{n + 1} } \\ &\overset{(a)}{\le} \sum_{i = 0}^m \ff{ B^i } \\ 
    &\overset{(b)}{\le} \ff{ A_{n + 1} } \sum_{i = 0}^\infty k^{-\ell i} \\ 
                                                                     &\overset{(c)}{=}  \ff{ A_{n + 1} } \left( \frac{1}{1 - k^{-\ell }} \right), 
  \end{align*} 
where inequality (a) follows from submodularity and nonnegativity of $f$,
  inequality (b) follows Claim \ref{claim:deletion}, and inequality (c) follows
  from the sum of a geometric series.
\end{proof}
Next, we bound the value of $\opt$ in terms of $\ff{ A_{n + 1} }$. 
\begin{lemma} \label{lemm:Agood}
  $\left(2 + \frac{1}{k^\ell - 1} \right)\ff{ A_{n + 1} } \ge \opt$.
\end{lemma}
\begin{proof}
  Let $O \subseteq \mathcal U$ be an optimal solution of size $k$
  to \mon; for each $o \in O$, let $i(o)$ be the iteration in which $o$
  was processed. Then
  \begin{align*}
    f(O) - \ff{ A^* } &\le \ff{ O \cup A^* } - \ff{ A^* } \\ 
    &\le \sum_{o \in O \setminus A^*}\ff{ A^* + o } - \ff{ A^* } \\ 
    &\le \sum_{o \in O \setminus A^*} \ff{ A_{i(o)} } / k \\ 
    &\le \sum_{o \in O \setminus A^*} \ff{ A_{n + 1} } / k \le \ff{ A_{n + 1} }, 
  \end{align*}
  by monotonicity and submodularity of $f$, the condition of Line \ref{line:add}, Lemma \ref{lemm:increasing}, and the size of $O$.
  From here, the result follows from Lemma \ref{lemm:aplus}.
\end{proof}
Recall that \algOneone returns the set $A'$, the last $k$ elements added to $A$. 
Lemma \ref{lemm:Aprime} shows that $2f( A' ) \ge \ff{A_{n + 1}}$.
\begin{lemma} \label{lemm:Aprime}
  $\ff{ A_{n + 1} } \le 2 \ff{A'}$.
\end{lemma}
\begin{proof}
  If $|A_{n + 1}| \le k$, $\ff{ A' } \ge \ff{A_{n + 1}}$ by monotonicity, and the lemma holds. 
  Therefore, suppose $|A_{n + 1}| > k$. Let $A' = \{a'_1, \ldots, a'_k \}$, in the order these elements
  were added to $A$. Let $A'_i = \{ a_1', \ldots, a'_i \}$,
  $A'_0 = \emptyset$. 
  Then
  \begin{align*}
    \ff{A'} &\ge \ff{A_{n + 1}} - \ff{A_{n + 1} \setminus A'} \\ 
            &= \sum_{i = 1}^k \left[ \ff{ (A_{n + 1} \setminus A') \cup A'_{i-1} + a'_i } \right. \\ &\qquad \left. - \ff{ (A_{n + 1} \setminus A') \cup A'_{i - 1} } \right] \\
    &\overset{(a)}{\ge} \sum_{i=1}^k \frac{\ff{(A_{n + 1} \setminus A') \cup A'_{i-1}}}{k} \\ 
    &\overset{(b)}{\ge} \sum_{i=1}^k \frac{\ff{A_{n + 1} \setminus A'}}{k} = \ff{A_{n + 1} \setminus A'}, 
  \end{align*}
  where inequality (a) is by the condition on Line \ref{line:add}, and
  inequality (b) is from monotonicity of $f$. 
  Thus $f(A_{n + 1}) \le \ff{A_{n + 1} \setminus A'} + \ff{A'} \le 2f(A'). \qedhere$
\end{proof}
Since $k \ge 2$,
Lemmas \ref{lemm:Agood} and \ref{lemm:Aprime} show that the set $A'$ of
\algOneone satisifes $f(A') \ge \left( \frac{1}{4 + 2/(k^\ell - 1)} \right) \opt$.
By the choice of $\ell$, $f(A') \ge \ratioOne \opt$. \end{proof}

\subsection{Post-Processing: $\qs_c$++ } \label{sec:qs-post}
In this section, we describe a simple post-processing procedure to improve the
objective value obtained by $\qs_c$. At the termination of the stream, 
$\qs_c$ stores a set $A$ of size $O(k \log k)$ from which the set $A'$
and solution are extracted, on which the worst-case approximation ratio
is proven in the previous section. However, the set $A$ may be regarded
as a filtered ground set of size $O(k \log k) \le n$, 
upon which any algorithm may be run to extract
a solution. As long as the post-processing algorithm has query and time complexity
and runtime $O(n)$, Theorem \ref{thm:qs} still holds for the resulting
single-pass streaming algorithm with post-processing. This modification
of $\qs$ is termed $\qs$++.

We remark that 
the condition of Line \ref{line:add} of $\qs_c$ may be changed to the following condition:
$\ff{A \cup C} - \ff{A} \ge \delta \ff{A} / k,$
for input parameter $\delta > 0$. In this case, it is not difficult to
 extend the analysis in the previous
section to show that the algorithm achieves ratio $[c(1 + \delta )(1 + 1/\delta)]^{-1},$
in memory $O( k \log k )$ and the same query complexity and runtime. This ratio is optimized
for $\delta = 1$, but when using post-processing with $\qs_c$++, smaller values of $\delta$
result in larger sets $A$, although still bounded in $O(k \log k) \le n$. 
We found 
in our empirical evaluation in Section \ref{sec:exp} that setting $\delta = c/10$ for $\qs_c$++ 
yields good empirical results.

\subsection{Lower Bound on Query and Time Complexity} \label{apx:thm2}
While it is clear that at least $n$ queries are required for any
constant factor if the algorithm is only allowed to query feasible sets (consider $k = 1$),
our algorithms bypass this restriction. 
Our next result is a lower bound on the number of queries (and hence also
the time complexity)
required to obtain a constant-factor approximation.
\begin{theorem} \label{thm:lower}
  Let $c \ge 2$ be an integer, and let $\epsi > 0$. Any (randomized) approximation algorithm
  for \mon with ratio $1 / c + \epsi$ for \mon with probability $\delta$
  requires at least $\lceil \delta n / (ck - 1) \rceil$ oracle queries and
  hence $\Omega (n / k)$ time. 
\end{theorem}
Theorem \ref{thm:lower} implies no constant-factor approximation
exists with $o(n)$ time in the value query model.
Another consequence of Theorem \ref{thm:lower} is that any 
algorithm with ratio $(1/2 + \epsi)$ with probability greater than $1 - 1/n$ 
requires at least $n$ queries. 

\begin{proof}
We prove the theorem for instances of SMCC with cardinaity constraint $k \ge 1$. Let $c \in \nats$, $c \ge 2$, and let $0 < \epsi < 1$.
Let $n \in \nats$, and let
$\mathcal U_n = \{ 0, 1, \ldots, n - 1 \}$. Define $f: 2^{\mathcal U_n} \to \reals$
by $f(A) = \min \{ |A|, ck \}$, for $A \subseteq \mathcal U_n$.
Next, we define a function $g$ that is hard to distinguish from $f$:
pick $a \in \mathcal U_n$ uniformly randomly. Let $g(A) = f(A)$ if
$a \not \in A$, and $g(A) = ck$ otherwise. Clearly, both $f$ and $g$
are monotone and submodular. 

Now, consider queries to $f$ and $g$ of a set $A \subseteq \mathcal U_n$. These
queries can only distinguish between $f$ and $g$ if $|A| \le ck - 1$ and $a \in A$;
in any other case, the values of $f(A)$ and $g(A)$ are equal. 
Consider a (possibly adaptive) sequence
of queries of sets $A_1,A_2,\ldots,A_m$. Without loss of generality,
we may assume $|A_i| \le ck - 1$ for each $i$, since the query
of any set of larger size yields no information about the
element $a$. Then the algorithm can correctly
distinguish $f$ from $g$ iff $a \in \bigcup A_i$, which happens
with probability at most $m(ck - 1)/n$, since $\left| \bigcup A_i \right| \le m(ck - 1)$.
Therefore, to distinguish between $f$ and $g$ with probability
at least $\delta$ requires at least $\lceil \delta n / (ck - 1) \rceil$ queries.

Since any approximation algorithm with ratio at least $1/c + \epsi$ with probability $\delta$
would distinguish between $f,g$ with probability $\delta$, since
the optimal solution with $f$ has value $k$, while $g(a) = ck$,
the theorem is proven. \qedhere
\end{proof}

\section{Multi-Pass Streaming Algorithm to Boost Constant Ratio to $1 - 1/e - \epsi$} \label{sec:brlarge}
\begin{algorithm}[t]
   \caption{A procedure to boost to from constant ratio $\alpha$ to ratio $1 - e^{-1 + \epsi}$ in $O(1/\epsi)$ passes, $1$ query per element per pass, and $O(k)$ memory. } \label{alg:boost}
   \begin{algorithmic}[1]
     \Procedure{\algTwo}{$f, k, \alpha, \Gamma, \epsi$}
     \State \textbf{Input:} evaluation oracle $f:2^{\mathcal N} \to \reals$, constraint $k$, constant $\alpha$, value $\Gamma$ such that $\Gamma \le \opt \le \Gamma / \alpha$, and $0 < \epsi < 1$.
     \State $\tau \gets \Gamma / (\alpha k)$, $A \gets \emptyset$.
     \While{ $\tau \ge (1 - \epsi )\Gamma / (4k)$ }
     \State $\tau \gets \tau ( 1 - \epsi )$
     \For{ $n \in \uni$ }\label{line:pass}
     \If{$f( A + n ) - f(A) \ge \tau$}\label{line:marge1}
     \State $A \gets A + n$\label{line:marge2}
     \EndIf
     \If{ $|A| = k$ }
     \State \textbf{return} $A$
     \EndIf
     \EndFor     
     \EndWhile
     \State \textbf{return} $A$
     \EndProcedure
\end{algorithmic}
\end{algorithm}
In this section, we describe $\algTwo$ (Alg. \ref{alg:boost}),
which given any $\alpha$-approximation $\mathcal A$ for \mon can
boost the ratio to $1 - e^{-1 + \epsi} \ge 1 - 1/e - \epsi$ using the output of $\mathcal A$. 
Theorem \ref{thm:br} is proven below.
\begin{theorem} \label{thm:br}
  Let $0 < \epsi < 1$. Suppose a deterministic 
  $\alpha$-approximation $\mathcal A$ exists for
  \mon. Then algorithm $\algTwo$ is a multi-pass streaming algorithm that when
  applied to the solution 
  of $\mathcal A$ yields a solution within factor $1 - e^{-1 + \epsi} \ge 1 - 1/e - \epsi$ 
  of optimal                    
  in at most $n ( \log( 4 / \alpha ) / \epsi + 1$ queries,
  $\log( 4 / \alpha ) / \epsi + 1$ passes, and $\oh{k}$ memory.
\end{theorem}
If the algorithm $\mathcal A$ is 
the algorithm provided by Theorem \ref{thm:online},
this establishes Theorem \ref{thm:optimal-det}.

\begin{figure*}[!ht]
  \subfigure[Objective, small $k$]{ \label{fig-sp:v-Google-smallk}
    \includegraphics[width=0.30\textwidth,height=0.15\textheight]{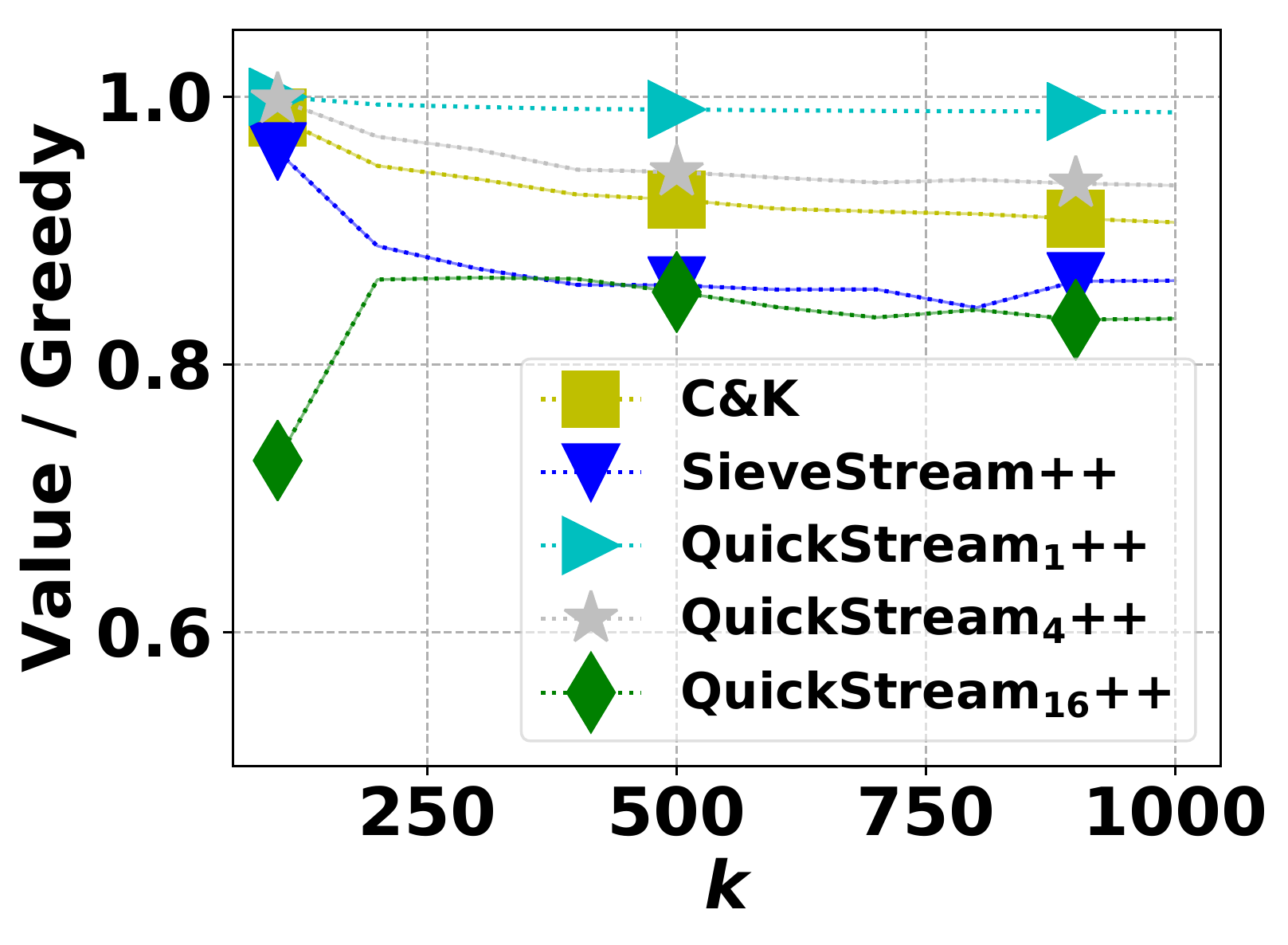}
  }
  \subfigure[Queries, small $k$]{ \label{fig-sp:q-Google-smallk}
    \includegraphics[width=0.30\textwidth,height=0.15\textheight]{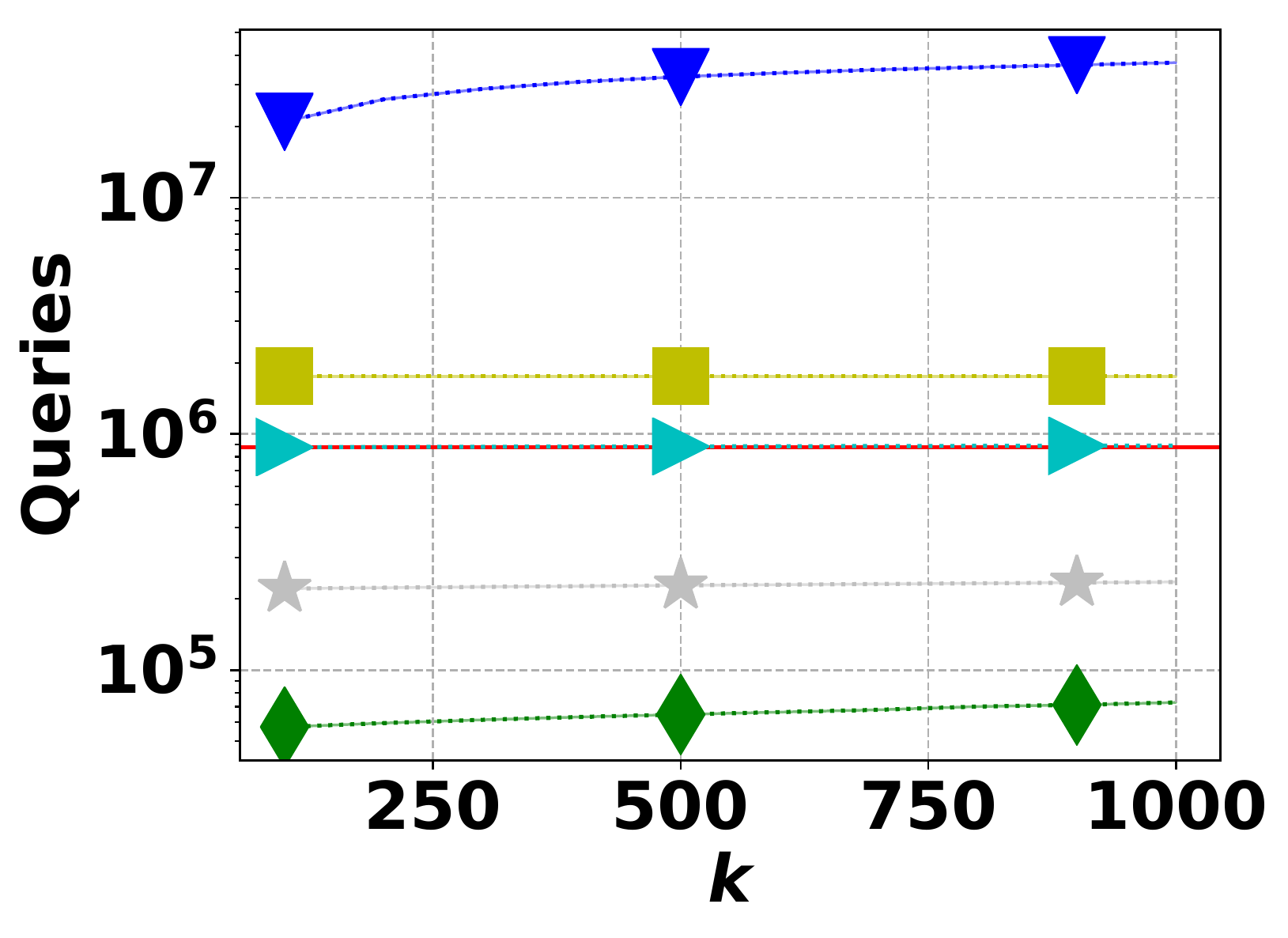}
  }
  \subfigure[Memory, small $k$]{ \label{fig-sp:m-Google-smallk}
    \includegraphics[width=0.30\textwidth,height=0.15\textheight]{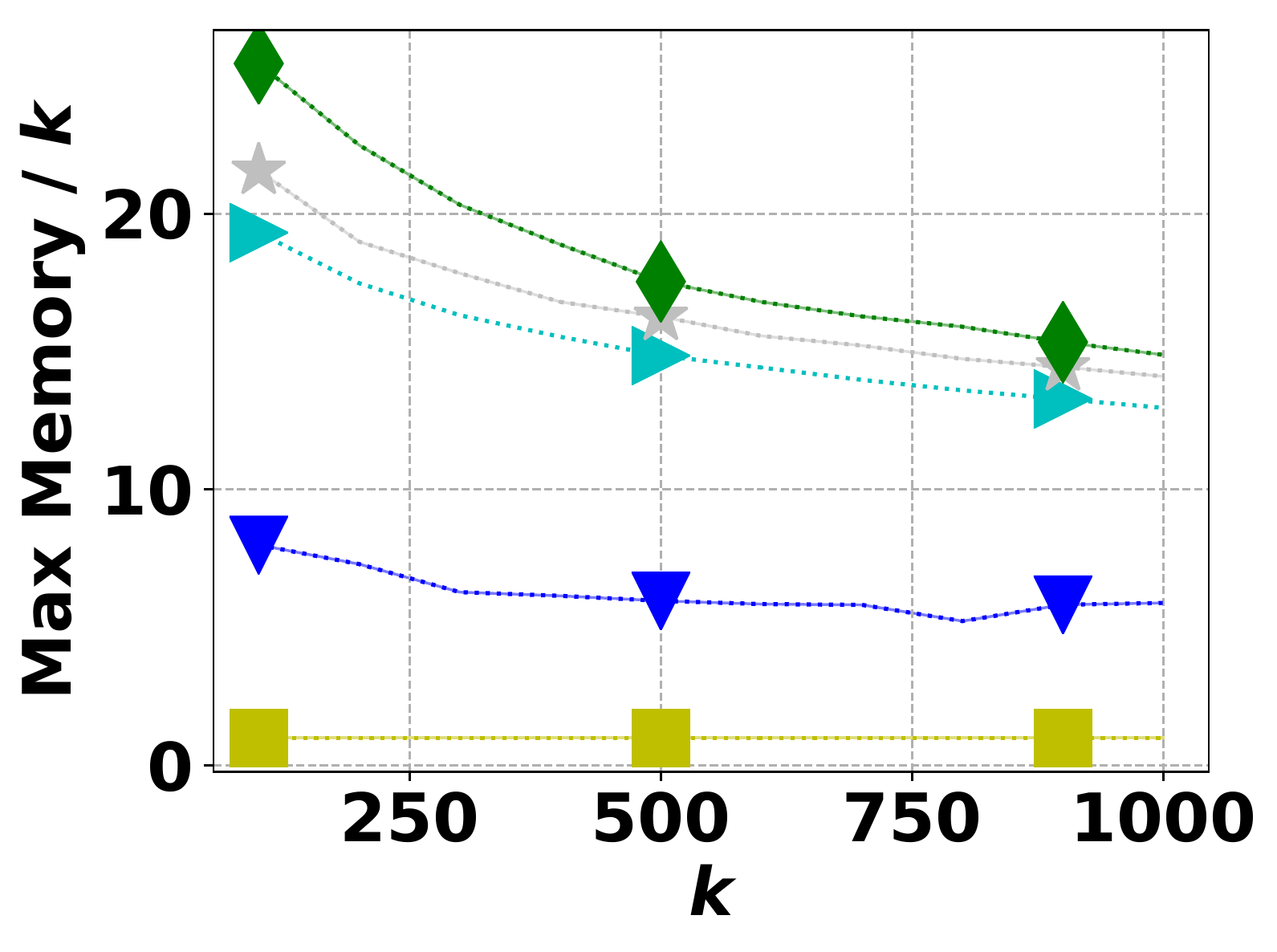}
  }

  \subfigure[Objective, large $k$]{ \label{fig-sp:v-Google-largek}
    \includegraphics[width=0.30\textwidth,height=0.15\textheight]{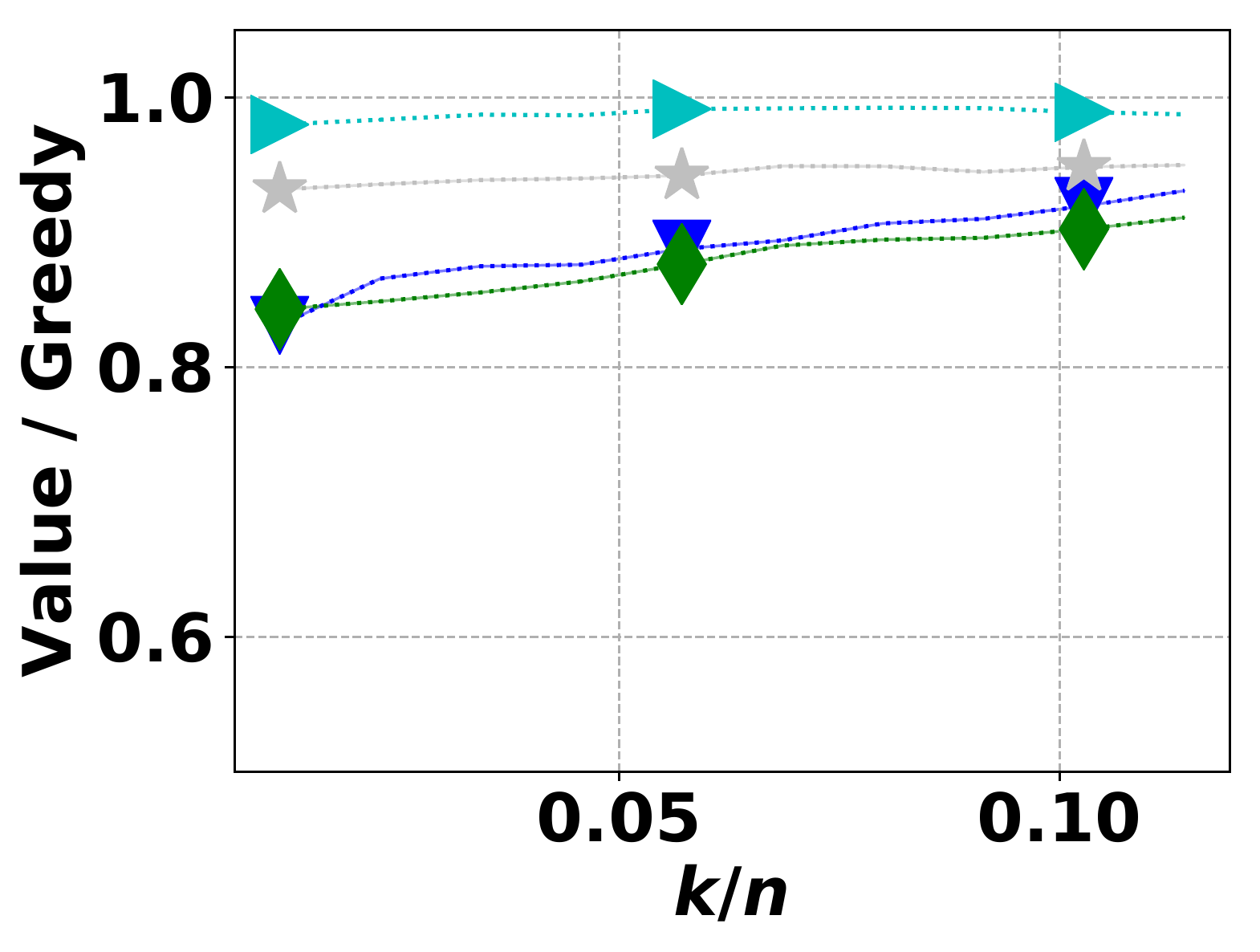}
  }
  \subfigure[Queries, large $k$]{ \label{fig-sp:q-Google-largek}
    \includegraphics[width=0.30\textwidth,height=0.15\textheight]{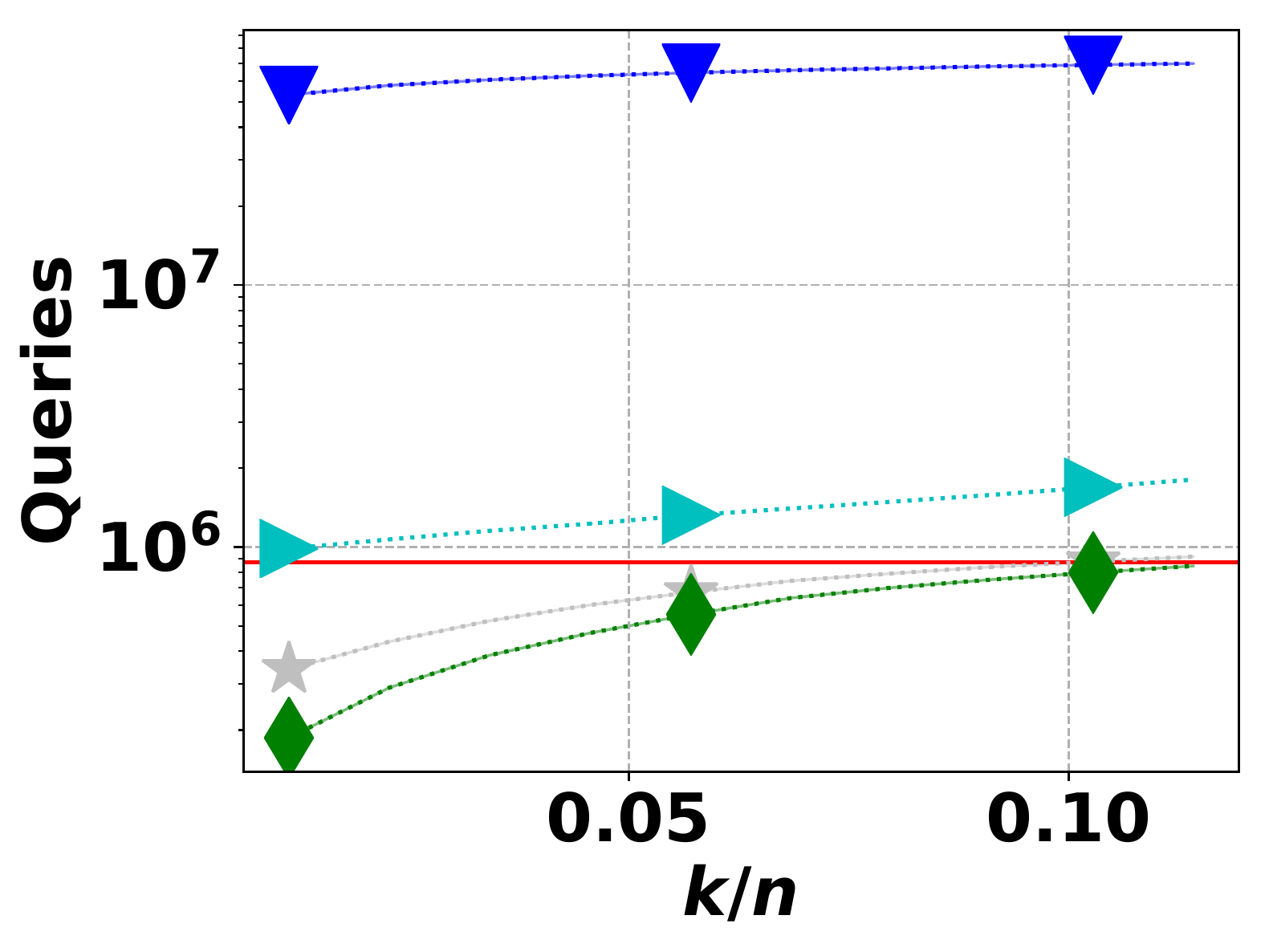}
  }
  \subfigure[Memory, large $k$]{ \label{fig-sp:m-Google-largek}
    \includegraphics[width=0.30\textwidth,height=0.15\textheight]{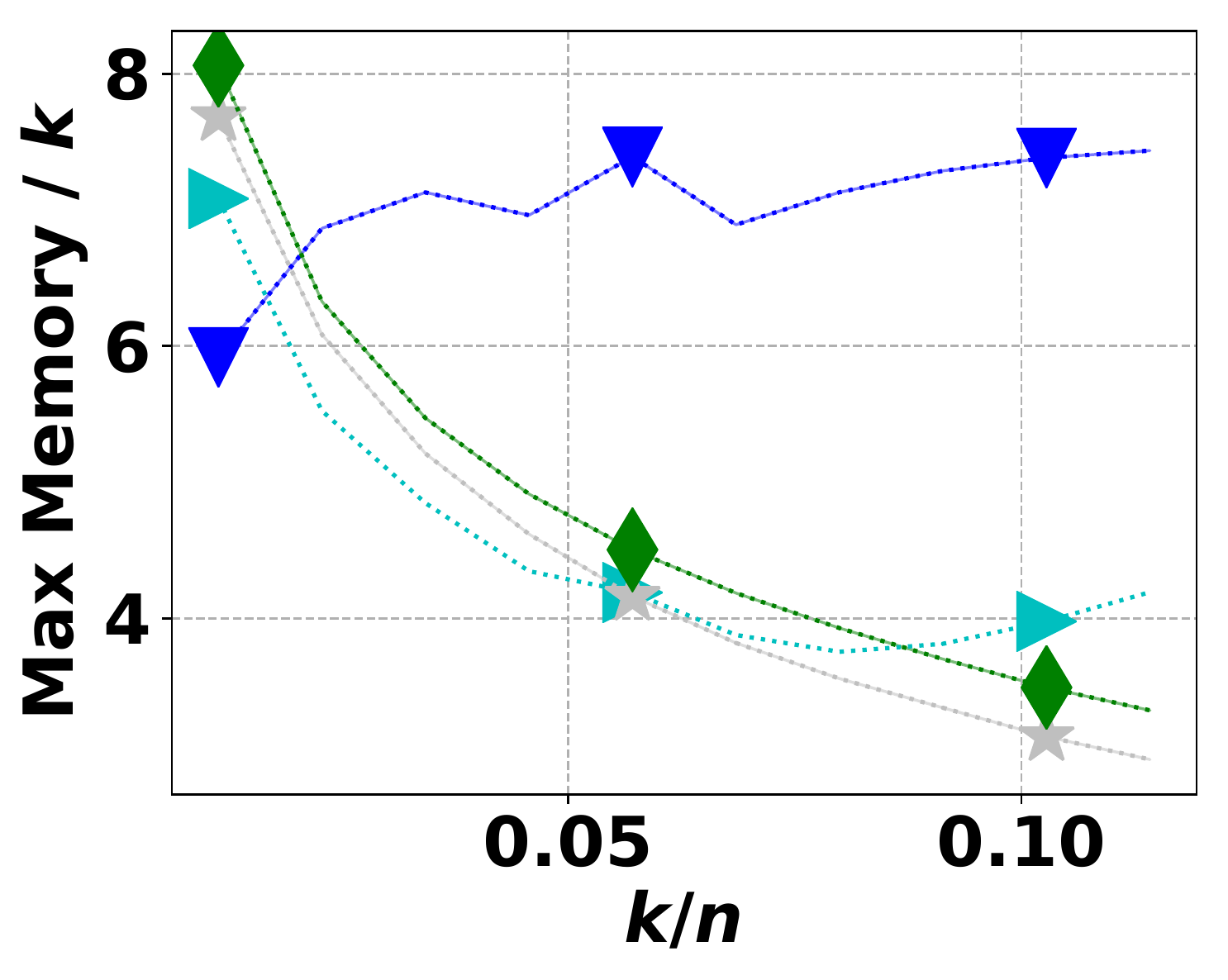}
  }
  \caption{Evaluation of single-pass streaming algorithms on web-Google $(n=875713)$, in terms of objective value normalized by the standard greedy value, total number of queries, and the maximum memory used by each algorithm normalized by $k$. The legend shown in \textbf{(a)} applies to all subfigures. }
  \label{fig:single-pass}
\end{figure*}
\begin{figure*}[t]
  \subfigure[]{ \label{fig:mp-v-google-smallk}
    \includegraphics[width=0.30\textwidth,height=0.15\textheight]{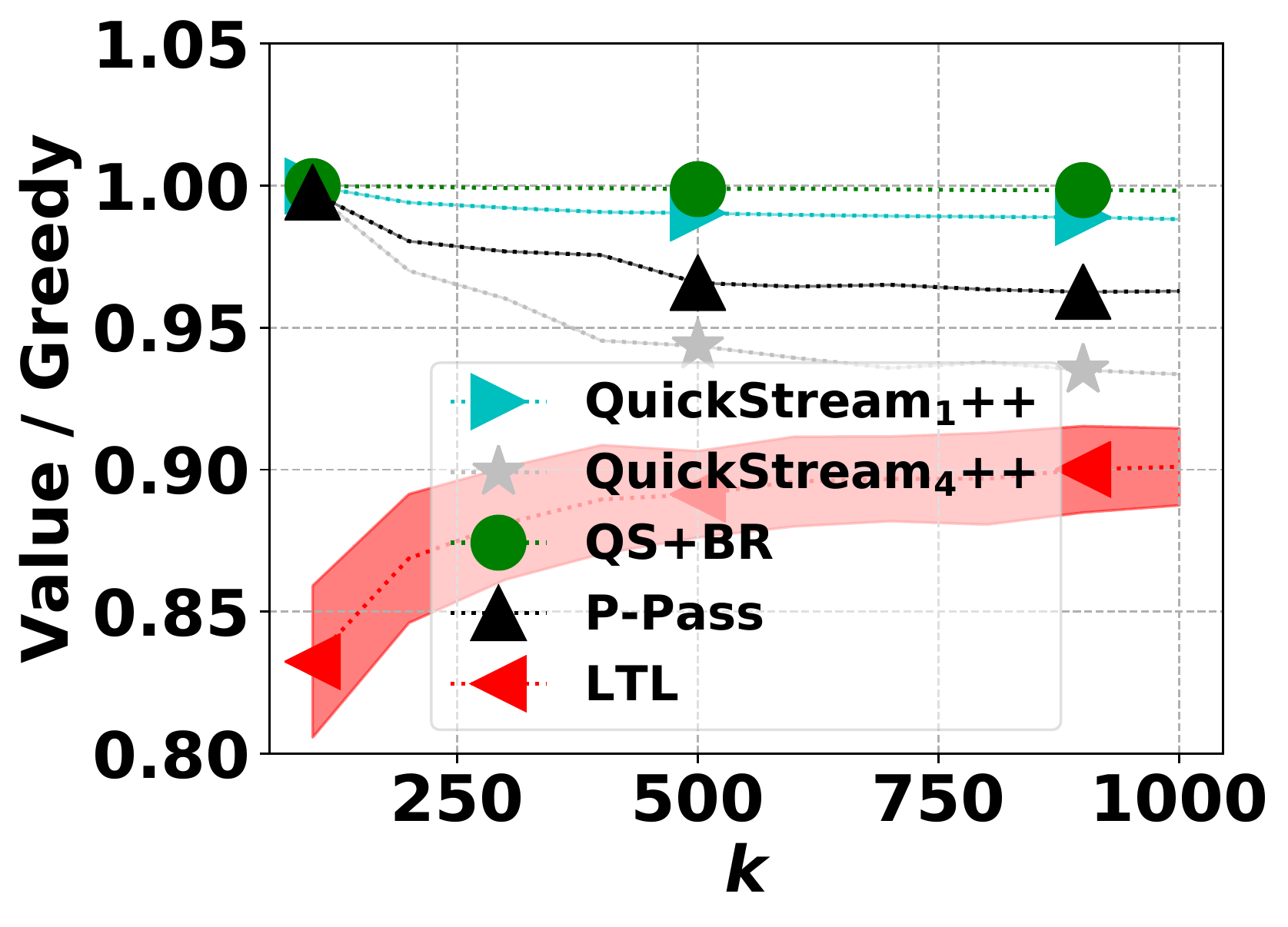}
  }
  \subfigure[]{ \label{fig:mp-q-google-smallk}
    \includegraphics[width=0.30\textwidth,height=0.15\textheight]{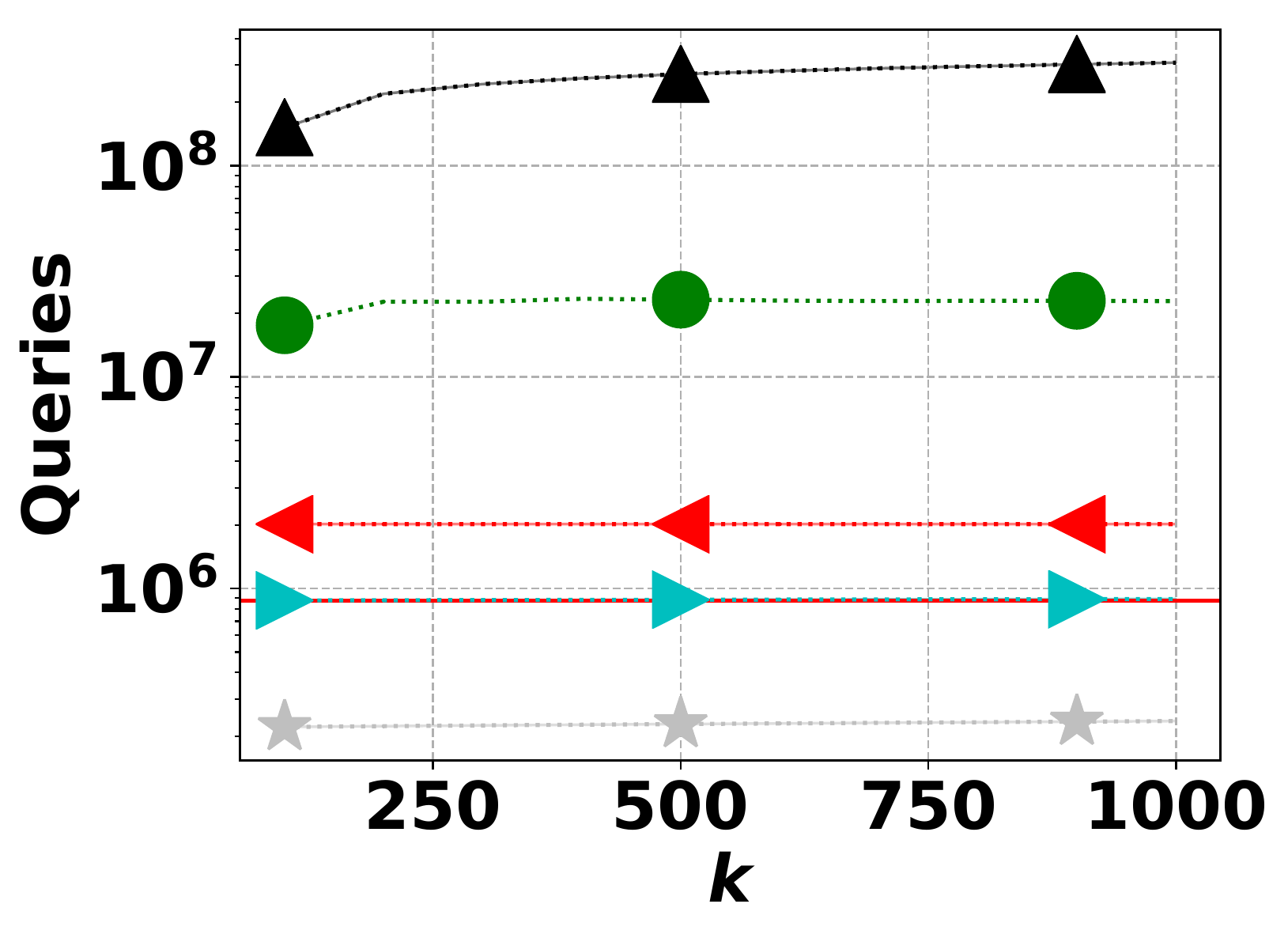}
  }
  \subfigure[]{ \label{fig:mp-v-google-largek}
    \includegraphics[width=0.30\textwidth,height=0.15\textheight]{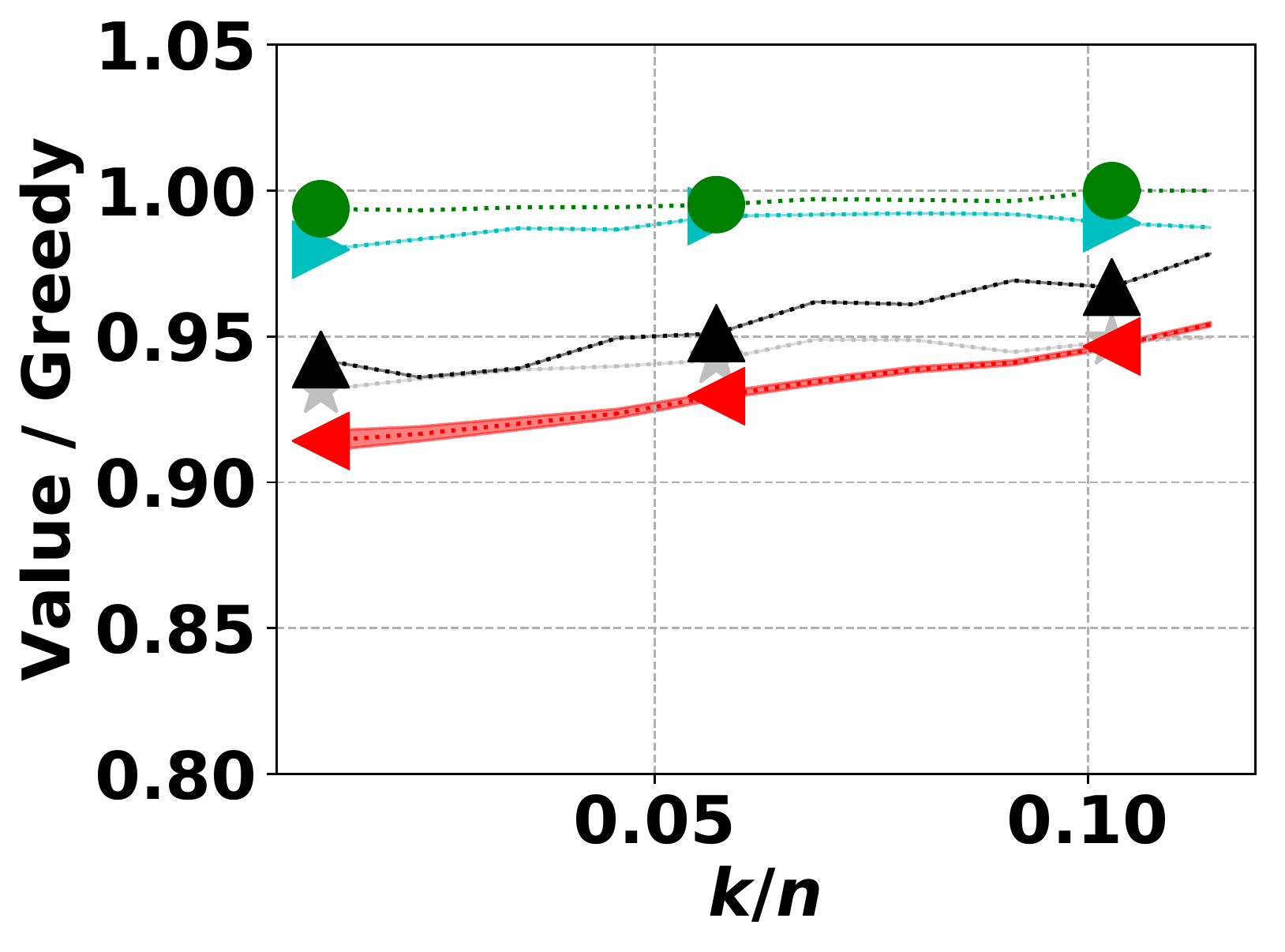}
  }
  \caption{Evaluation of our algorithms compared with the multi-pass \ppass and non-streaming algorithm \lazy. We compare the objective value (normalized by the standard \greedy objective value) and total queries on web-Google for the maximum cover application for both small and large $k$ values. The large $k$ values are given as a fraction of the number of nodes in the network. The legend shown in \textbf{(a)} applies to all subfigures. } \label{fig:main}
\end{figure*}

As input, the algorithm $\algTwo$
takes an instance $(f,k)$ of \mon, an approximate solution value $\Gamma$, and accuracy parameter $\epsi > 0$. On the instance $(f,k)$, it must hold that $\Gamma \le \opt \le \Gamma / \alpha$, where $\opt$ is the value of an optimal solution.
The algorithm works by making one pass (line \ref{line:pass})
through the ground set for each threshold value $\tau$, during which
any element with marginal gain at least $\tau$ to $A$ is added to $A$
(lines \ref{line:marge1} -- \ref{line:marge2}).
The maximum
and minimum values of $\tau$ are determined by $\Gamma, \alpha$, and $k$: initally $\tau = \Gamma / (\alpha k)$,
and the algorithm terminates if $\tau < (1 - \epsi) \Gamma / (4k)$; 
each iteration of the \textbf{while} loop,
$\tau$ is decreased by a factor of $(1 - \epsi )$.
The set $A$ is initially empty;
if $|A| = k$, the algorithm terminates and returns $A$; otherwise, at most
$O( \log (1 / \alpha ) / \epsi )$ passes are made until the minimum threshold value is reached.

Intuitively, the $1 - 1/e - \epsi$ ratio is achieved since the $\alpha$-approximate solution $\Gamma$ allows the algorithm to approximate the value for $\tau$ of $\opt / k$ in a constant number of guesses. Once this threshold has been reached, only $\log(1/4) / \epsi$ more values of $\tau$ are needed to achieve the desired ratio. While \algTwo may be used with
any $\alpha$-approximation, if it is used with $\qs_1$,
the resulting algorithm is the first linear-time, deterministic, $(1 - 1/e - \epsi)$-approximation
for SMCC, which is a multi-pass streaming algorithm.

\begin{proof}[Proof of Theorem \ref{thm:br}]
  Suppose $0 < \epsi < 1$.
  Let $(f,k)$ be an instance of \mon. The algorithm 
  is to first run $\mathcal A$,
  to obtain set $A'$. Next,
  \algTwo is called with parameters $(f, k, \alpha, f(A'), \epsi )$. 
  Observe that the inital value of
  the threshold
  $\tau$ in the \textbf{while} loop is at least $(1 - \epsi)\opt / k$, and the
  final value of $\tau$ is at most $\opt / (4k)$.
  
  Consider the case that at termination $|A| < k$. Then by the last iteration of the \textbf{while} loop,
  submodularity and monotonicity of $f$, 
  \begin{align*}
    f(O) - f(A) &\le f( O \cup A ) - f(A)  \\
    &\le \sum_{o \in O \setminus A} f( A \cup \{ o \} ) - f(A) \\
    &\le \sum_{o \in O \setminus A} \Gamma / (4k) \le \opt / 4,
  \end{align*}
  from which $f(A) \ge 3 \opt / 4 \ge ( 1 - e^{-1 + \epsi} ) \opt$.

  Next, consider the case that at termination $|A| = k$. Let $A_i = \{ a_1, a_2, \ldots, a_i \}$, ordered
  by the addition of elements to $A$, and let $A_0 = \emptyset$. 
  \begin{claim} \label{claim:basic} Let $i \in \{0, \ldots, k - 1 \}$. Then
    $$\ff{ A_{i + 1} } - \ff{ A_i } \ge \frac{(1 - \epsi)}{k}\left( \opt - \ff{ A_i } \right)$$
    \end{claim}
    \begin{proof}
      Let $i \in \{0, \ldots, k - 1 \}$. First, suppose $a_{i + 1}$ is added to $A_i$ during
      an iteration with
      $\tau \ge (1 - \epsi )\opt / k$. In this case, 
      $f( A_{i + 1} ) - f(A_i) \ge \tau \ge (1 - \epsi)\opt / k \ge \frac{(1 - \epsi)}{k}( \opt - f( A_i ) )$. 

      Next, suppose $a_{i + 1}$ is added to $A_i$ during an iteration with $\tau < ( 1 - \epsi )\opt / k$. 
      Consider the set $O \setminus A_i$; in the previous iteration of the \textbf{while} loop, 
      no element of $O \setminus A_i$ is added to $A$; hence, by submodularity,
      for all $o \in O \setminus A_i$, $f( A_i + o ) - f(A_i) < \tau / (1 - \epsi)$. Therefore,
      \begin{align*}
        f( A_{i + 1} ) - f( A_i ) &\ge \tau \\
                                  &\ge \frac{(1 - \epsi)}{k} \sum_{o \in O \setminus A_i} f( A_i \cup \{ o \} ) - f( A_i) \\
        &\ge \frac{(1 - \epsi)}{k} ( f( O \cup A_i) - f(A_i) ) \\
        &\ge \frac{(1 - \epsi)}{k} ( \opt - f( A_i ) ).\qedhere
      \end{align*}
    \end{proof}
    From Claim \ref{claim:basic}, standard arguments show the $f(A_k) \ge \opt \left( 1 - e^{-1 + \epsi} \right) \ge \opt (1 - 1/e - \epsi)$.

    For the query complexity, observe that the \textbf{for} loop of \algTwo
    makes at most $n$ queries, and the \textbf{while}
    loop requires $\log ( \alpha / 4 ) / \log ( 1 - \epsi ) + 1 \le \log (4 /\alpha ) / \epsi + 1$ iterations. 
\end{proof}


\section{Empirical Evaluation} \label{sec:exp}
In this section, we demonstrate that the objective value achieved empirically by 
our algorithm $\qs_c$++ beats that of
the state-of-the-art algorithms \lazy, \sstream, and \ck,
while using the fewest queries and only a single pass.
Our multi-pass algorithm \qsbr ($\qs_1$ followed by $\algTwo$) achieved
mean objective value 
better than $0.99$ of the 
standard \greedy value across all instances tested.
\paragraph{Algorithms}
Our algorithms are compared to the following
methods:
\greedy, the standard greedy algorithm analyzed by
\citet{Nemhauser1978}, \lazy \citep{Mirzasoleiman2014},
\sstream \citep{Kazemi2019}, \ppass \citep{Jakub2018},
and \ck \citep{Chakrabarti2015},
as described in Section \ref{sec:intro}.
Randomized algorithms
were averaged over $10$ independent runs and the shaded regions in
plots correspond to one standard deviation.
Any algorithm
with an accuracy parameter $\epsi$ is run with $\epsi = 0.1$ unless otherwise
specified.

We evaluate our algorithm $\qs_c$++ for various values of $c$. The post-processing
procedure run on $A$ is taken to be our linear time \algTwo 
and we set parameter $\delta = c / 10$ (see Section
\ref{sec:qs-post} for the definition of $\delta$). We also evaluate our 
multi-pass algorithm \qsbr.

\paragraph{Applications}
We evaluate all of the algorithms on two applications of \mon: the first is 
maximum coverage on a graph: for each set of vertices $S$, the value of $f(S)$
is the number of vertices adjacent to the set $S$. The second application
is the revenue maximization problem on a social network \citep{Hartline2008},
 a variant of influence maximization. For detailed specification 
of these applications, see Appendix \ref{apx:exp}.
We evaluate on a variety of network technologies from
the Stanford Large Network Dataset Collection \citep{snapnets}, including 
ego-Facebook ($n = 4039$) and web-Google ($n=875713$),
among others listed in Appendix \ref{apx:exp}. Values of $k$ evaluated
include small values ($k \le 1000$) and large values $(k = \Omega(n))$.

\paragraph{Results: Single-Pass Algorithms}
In Fig. \ref{fig:single-pass}, representative
results are shown for the single-pass algorithms.
Results were qualitatively similar across applications and datasets;
additional results are shown in Appendix \ref{apx:exp}.

\paragraph{Objective Value} For small $k$ ($k \le 1000)$, the mean objective value (normalized by 
the standard \greedy value) obtained by
each single-pass algorithm across all instances is as follows: 
$\qs_1$++ 0.99; $\qs_4$++ 0.95; \ck 0.93; \sstream 0.87; $\qs_{16}$++ 0.84.
On the instances with large $k$ ($k \le 0.1n$), the means are: 
$\qs_1$++ 0.99; $\qs_4$++ 0.94; \sstream 0.89; $\qs_{16}$++ 0.88.

\paragraph{Queries} In terms of queries, $\qs_{c}$++ required roughly $n/c$ queries
for small $k$; the the next smallest was \ck, which required
$2n$ queries, followed by \sstream, which started at more than $10n$ queries
and increased logarithmically with $k$. 
For large $k$, the queries of $\qs_{c}$++ increased due to the
$O(n)$ post-processing step which depends on $k$, but always remained less than $2n$.

The algorithm \ck, while very efficient in terms of queries,
was unable to run in a reasonable timeframe on our larger instances. 
Most of the algorithms we evaluate (including both of our algorithms)
use a marginal gain query of sets that only increase in size,
which yields an optimized implementation for
the maximum cover application. However, \ck cannot be implemented with this optimization
and requires the full $O(n)$ oracle query; thus, on some instances we were
able to run the standard greedy algorithm but not \ck. This illustrates the
fact that the oracle query complexity only constitutes partial information about the
runtime of the algorithm. 

\paragraph{Memory} As shown in Figs. \ref{fig-sp:m-Google-smallk} and
\ref{fig-sp:m-Google-largek}, the memory usage of the algorithms
remained at most a constant times $k$; for $\qs_c$++, this constant
decreased as $k$ increased, and with large enough $k$, the algorithms used
less memory than \sstream. In terms of memory, \ck is optimal
both theoretically and in practice, as it stores only $k$ elements.

\paragraph{Results: Multi-Pass and Non-Streaming Algorithms}
In Fig. \ref{fig:main}, we show results of our algorithms
$\qs_c$++ and $\qsbr$, in comparison with the multi-pass
\ppass algorithm and the non-streaming \lazy algorithm
on web-Google. Surprisingly, our single-pass algorithm $\qs_1$++
beats the objective values of both \ppass and \lazy, as it obtained $0.99$ of the standard greedy
value on average across all instances (both
small and large $k$). The only algorithm
with better objective value than $\qs_1$++
is our multipass \qsbr. The algorithm $\qs_4$++
exceeded the objective value of \lazy despite
using $1/8$ of the queries. 

\section{Conclusions}
In this work, we have provided the first constant-factor
algorithms for SMCC that make a linear number of oracle queries
and arithmetic operations.  Supplemented with post-processing
heuristics, our single-pass
algorithm \qs achieves state-of-the-art empirical
objective value while
using fewer than $n$ queries of the objective function.
Our multi-pass algorithm \qsbr
nearly achieves the optimal worst-case ratio of $1-1/e$ and is the 
first deterministic algorithm to do so with linear query complexity.

\section{Acknowledgments}
The work of A. Kuhnle was partially supported by Florida State University. We thank Victoria G. Crawford and the anonymous reviewers
for helpful feedback on earlier versions of the manuscript.

\bibliographystyle{plainnatfixed}
\bibliography{mend}

\clearpage

\appendix
\onecolumn

\section{Additional Related Work} \label{apx:rw}

\paragraph{Online Algorithms} A more restrictive streaming model is the 
preemptive, online model proposed by \citet{Buchbinder2015b}.
In this setting, 
the algorithm
receives elements one by one in an arbitrary order and 
must maintain a competitive solution with respect to the optimal
solution on elements seen so far; the algorithm is allowed to discard elements that were
previously chosen into the solution and must maintain a feasible
solution (a set of size at most $k$).
\citet{Buchbinder2015b} described a deterministic $1/4$-competitive algorithm in
this model that requires $O(kn)$ queries.
\citet{Chan2017} improved the competitive ratio to $0.296$ for
a deterministic algorithm in $O(kn)$ queries; their ratio 
converges to $\approx 0.318$ as $k \to \infty$. 
They also show that the ratio of $0.318$ is optimal in this
online model.
Our algorithms 
are not online in this sense, since they maintain an
infeasible set of size $O(k \log k)$ rather than a feasible set
of size $k$ and if $c > 1$, $\qs_c$ requires additional processing
at termination of the stream. However, $\altqs_c$ requires no
processing at the end of the stream and does maintain a competitive ratio 
that converges to $\approx 0.316 / c$.

\section{Variants of \algOne} \label{apx:qs}
In this section, we describe algorithms that are similar in design to
$\qs_c$. 
In Section \ref{apx:qss}, we describe $\qss_c$, designed for
the case $k = 1$. Finally, in Section \ref{apx:altqs}, we
describe $\altqs_c$, designed to have an improved ratio for
$k \ge 8c/e$. 

Observe that Theorem \ref{thm:online} is a direct consequence of
Theorems \ref{thm:qs}, \ref{thm:qss}, and \ref{thm:altqs}.



\subsection{The $\qss_c$ Algorithm} \label{apx:qss}
In this section, we describe the algorithm $\qss_c$, 
a deterministic, single-pass algorithm that
has guarantees summarized in the following theorem.
Full pseudocode is given in Alg. \ref{alg:qss}.
After receipt of $c$ elements stored in buffer $C$,
the algorithm evaluates $f(C)$ and replaces $A$ 
with $C$ if $f(C) > f(A)$. At termination, the
maximum singleton in $A$ is returned. 

\begin{theorem} \label{thm:qss}
The algorithm $\qss_c$ is a deterministic, single-pass
algorithm with ratio $1/c$ if $k = 1$,
query complexity $\qsquery$,
and memory complexity $O(c)$.
\end{theorem}
\begin{proof}
Suppose $k = 1$. Observe that at termination of the algorithm
any singleton $u \in \mathcal U$ satisifes $f(u) \le f(A)$.
Further, at termination
of the stream, the element $a$ in $A$ maximizing $f$ is returned. Let $b$
be an optimal singleton; by submodularity and the fact $|A| \le c$, $cf(a) \ge f(A) \ge f(b)$. 

Memory complexity and query complexity are clear.
\end{proof}

\begin{algorithm}[t]
   \caption{For each $c \ge 1$, a single-pass algorithm with approximation ratio $1/c$ for SMCC if $k = 1$. The query complexity is  $\lceil n / c \rceil + c$, memory complexity is $O(c)$.} \label{alg:qss}
   \begin{algorithmic}[1]
     \Procedure{$\qss_c$}{$f, k$}
     \State \textbf{Input:} oracle $f$, cardinality constraint $k$
     \State $A \gets \emptyset$, $C \gets \emptyset$
     \For{ element $e$ received }
     \State $C \gets C + e$
     \If{$|C| = c$ or stream has ended}
     \If{$f(C) > f(A)$}
     \State $A \gets C$
     \EndIf
     \State $C \gets \emptyset$
     \EndIf
     \EndFor
     \State \textbf{return} $\argmax_{a \in A} f(a)$ \label{line:qssterm}
     \EndProcedure
\end{algorithmic}
\end{algorithm}
\subsection{The $\altqs_c$ Algorithm} \label{apx:altqs}
\begin{algorithm}[t]
   \caption{For each $c \ge 1$, a single-pass algorithm with approximation ratio $\altratio$ if $k \ge 8c/e$. The query complexity is  $\lceil n / c \rceil $.} \label{alg:altqs}
   \begin{algorithmic}[1]
     \Procedure{$\altqs_c$}{$f, k$}
     \State \textbf{Input:} oracle $f$, cardinality constraint $k$
     \State $A \gets \emptyset$, $A' \gets \emptyset$, $C \gets \emptyset$, $j \gets 0$
     \For{ element $e$ received }
     \State $C \gets C + e$
     \If{$|C| = c$ or stream has ended}
     \If{$f( A \cup C ) - f(A) \ge cf(A) / k$}\label{line:altadd}
     \State $A \gets A \cup C$
     \State $j \gets j + 1$
     \EndIf
     \If{$j > 6(k + 1) \log_2 (k)$}
     \State $A \gets \{ 3(k + 1)\log_2 (k) \text{ blocks most recently added to } A \}$\label{line:altdelete-A}
     \State $j \gets 3(k + 1) \log_2 (k)$
     \EndIf
     \State $C \gets \emptyset$
     \EndIf
     \State $A' \gets \{ k \text{ elements most recently added to } A \}$
     \EndFor
     \State \textbf{return} $A'$
     \EndProcedure
\end{algorithmic}
\end{algorithm}
In this section, we describe algorithms,
parameterized by $c$,  that require $\lceil n / c \rceil$
queries, have $\oh{ c k \log (k) }$ memory complexity, and have ratio
that converges to
$(1 - 1/e)/(1 + c)$ as $k \to \infty$. However, for small $k$, these algorithms
may not have any approximation ratio. We refer to these algorithms as $\altqs_c$.

Full pseudocode for $\altqs_c$ is given in Alg. \ref{alg:altqs}. The main differences
with $\qs_c$ are 1) a block $C$ is added to $A$ only if the gain exceeds $c f(A) / k$ rather
than $f(A) / k$ as in $\qs_c$; (2) $A'$ keeps only the last $k$ elements added,
rather than the last $k$ blocks; hence, there is no need to partition $A'$ 
at the end of the algorithm. Instead, the set $A'$ is simply returned. 
The rest of the section proves the following theorem.
\begin{theorem} \label{thm:altqs}
The algorithm $\altqs_c$ is a single-pass, deterministic streaming algorithm
with approximation ratio
$$\altratio ,$$
if $k \ge 8c/e$, query complexity $\lceil n / c \rceil $, and
memory complexity $O(ck \log( k ))$.
\end{theorem}
\begin{proof}
In addition to Claim \ref{claim:elementary} above,
we need the following elementary fact about the number $e$:
\begin{claim} \label{claim:elementaryc}
For any real number $x > 0$, $(1 + 1/x)^x < e < (1 + 1/x)^{x + 1}$.
\end{claim}
We will actually show that \altqs maintains a competitive ratio
with respect to the optimal solution on the elements seen thus far;
suppose $m$ blocks have been received, 
let $C_i$ denote the $i$-th block of elements
processed on line \ref{line:altadd}.
Let $\opt_{\mathcal N}$ denote
the optimal
solution to SMCC with input $(f\restriction_{\mathcal N},k)$, 
where $\mathcal N = \bigcup_{i = 1}^m C_i \subseteq \mathcal U$.
 Let $A_i$ denote the value of set $A$ immediately
before processing the $i$-th block $C_i$, and let $A_{m + 1}$
denote the value of $A$ after processing all blocks. Finally, let $A^*$ denote
$\bigcup_{i + 1}^{m + 1} A_i$.

The following two lemmas have exactly analogous proofs to Lemmas \ref{lemm:increasing} and \ref{lemm:aplus} by replacing blocks for elements, $2$ for $\ell$, and noting that $(1 + c/k) \ge (1 + 1/k)$.
We provide the proofs for completeness.
\begin{lemma} \label{lemm:increasingc}
  Suppose $k > 1$; let $1 \le i \le m$. Then
  $f(A_i) \le f(A_{i + 1})$.
\end{lemma}
\begin{proof}
  If no deletion is made during the processing of
  block $C_i$,
  then the change
  in $f(A)$ is clearly nonnegative. So suppose 
  deletion of set $B$ from $A$ occurs on line \ref{line:altdelete-A} during
  this iteration. Observe that $A_{i + 1} = (A_i \setminus B) \cup C_i$,
  because the deletion is triggered by the addition of block $C_i$ to $A_i$.
  In addition, at some iteration $j < i$ 
  of the \textbf{for} loop, it holds that $A_j = B$.
  From the beginning of iteration $j$ to the beginning of iteration $i$
  there have been $3(k + 1) \log_2( k ) - 1 \ge 2(k + 1) \log_2(k)$ 
  additions of blocks and no deletions to $A$, which add
  precisely the elements 
  in $(A_i \setminus A_j)$.

  It holds that
  $$\ff{ A_i \setminus A_j } \overset{(a)}{\ge} \ff{ A_i } - \ff{ A_j } \overset{(b)}{\ge} \left( 1 + \frac{1}{k} \right)^{2(k + 1) \log k} \cdot f( A_j ) - f (A_j) \overset{(c)}{\ge} (k^{2} - 1 ) f( A_j ),
  $$
  where inequality (a) follows from submodularity and nonnegativity of $f$,
  inequality (b) follows from the fact that each addition from $A_j$ to $A_i$ increases
  the value of $f(A)$ by a factor of at least $(1 + 1/k)$, and inequality (c)
  follows from Claim
  \ref{claim:elementary}.
  Therefore
  \begin{equation} \label{ineq:1c}
    f(A_i) \le \ff{A_i \setminus A_j} + \ff{ A_j } \le \left( 1 + \frac{1}{k^2 - 1} \right) \ff{A_i \setminus A_j}.
  \end{equation}
Next,
  \begin{equation} \label{ineq:2c}
    \ff{ (A_i \setminus A_j) \cup C_i} - \ff{ A_i \setminus A_j } \overset{(d)}{\ge} \ff{ A_i \cup C_i } - \ff{ A_i } \overset{(e)}{\ge} \ff{A_i} / k  \ge \ff{A_i \setminus A_j } / k, \end{equation}
  where inequality (d) follows from submodularity, and inequality (e) is by
  the condition to add $C_i$ to $A_i$ on line \ref{line:altadd}.
  Finally, using Inequalities (\ref{ineq:1c}) and (\ref{ineq:2c}) as indicated below, we have
  \begin{align*}
    \ff{ A_{i + 1}} =\ff{ A_i \setminus A_j \cup C_i } \overset{\text{By (\ref{ineq:2c})}}{\ge} \left( 1 + \frac{1}{k} \right) \ff{ A_i \setminus A_j } \overset{\text{By (\ref{ineq:1c})}}{\ge} \frac{1 + \frac{1}{k}}{1 + \frac{1}{k^{2} - 1}} \cdot f(A_i) \ge f( A_i ), 
  \end{align*}
  where the last inequality follows since $k \ge 2$.
\end{proof}
\begin{lemma} \label{lemm:aplusc}
  $$\ff{A^*} \le \left( 1 + \frac{1}{k^3 - 1} \right)\ff{ A_{m + 1} }. $$
\end{lemma}
\begin{proof}
  Observe that $A^* \setminus A_{m + 1}$ may be written as the
  union of pairwise disjoint sets, each of which is size
  $3c(k+1) \log_2(k) + 1$ and was deleted on line \ref{line:altdelete-A} of Alg. \ref{alg:altqs}.
  Suppose there were $l$ sets deleted from $A$; write 
  $A^* \setminus A_{m + 1} = \{ B^i : 1 \le i \le l \}$, where
  each $B^i$ is deleted on line \ref{line:delete-A}, ordered such that $i < j$
  implies $B^i$ was deleted after $B^j$ (the reverse order in which they were deleted); 
  finally, let $B^0 = A_{m + 1}$.

  \begin{claim} \label{claim:deletionc}
    Let $0 \le i \le l$.
    Then $\ff{B^i} \ge k^{3} \ff{B^{i+1}}$.
  \end{claim}
  \begin{proof}
    Let $B^i$, $B^{i + 1} \in \mathcal B$. 
    There are at least $3(k + 1) \log k + 1$ blocks added to $A$ and exactly one deletion
    event during the period between starting when $A = B^{i + 1}$ until $A = B^i$. 
    Moreover, each addition except possibly one (corresponding to the deletion event)
    increases $f(A)$ by a factor
    of at least $1 + 1/k$. Hence, by Lemma \ref{lemm:increasingc} and Claim \ref{claim:elementary},
  $\ff{B^i} \ge k^{3} \ff{B^{i+1}}$.
  \end{proof}
  By Claim \ref{claim:deletionc}, for any $0 \le i \le l$
  $\ff{ A_{m + 1} } \ge k^{3 i} \ff{ B^i }$.
  Thus, by submodularity and nonnegativity of $f$ and the 
  sum of a geometric series,
  \begin{align*}
    \ff{ A^* } \le \ff{ A^* \setminus A_{m + 1} } + \ff{ A_{m + 1} } &\le \sum_{i = 0}^m \ff{ B^i } \\
    &\le \ff{ A_{m + 1} } \sum_{i = 0}^\infty k^{-3 i}  \\
    &=  \ff{ A_{m + 1} } \left( \frac{1}{1 - k^{-3 }} \right). \qedhere
  \end{align*}
\end{proof}
The next lemma shows that $f(A_{m + 1} )$ has a significant fraction of the
optimal value.  
\begin{lemma} \label{lemm:Agoodc}
 $\left(1 + c + \frac{1}{k^3 - 1} \right)\ff{ A_{m + 1} } \ge \opt_\uni$.
\end{lemma}
\begin{proof}
  Let $O \subseteq \uni$ be an optimal solution to of size $k$
  to \mon. Let $C_o$ denote the block containing $o \in O$ that is 
  considered for addition into $A$.
  Then by monotonicity and submodularity of $f$, the fact that if block $C_i$
  is not added to $A$, $\ff{A \cup C_i} - \ff{A_i} < c\ff{A_i} / k$, and by Lemma
  \ref{lemm:increasingc}, we have
  \begin{align*}
    f(O) - \ff{ A^* } &\le \ff{ O \cup A^* } - \ff{ A^* } \\
                      &\le \sum_{o \in O \setminus A^*}\ff{ A^* \cup \{ o \} } - \ff{ A^* } \\
                      &\le \sum_{o \in O \setminus A^*}\ff{ A_o \cup \{ o \} } - \ff{ A_o } \\
                      &\le \sum_{o \in O \setminus A^*}\ff{ A_o \cup C_o } - \ff{ A_o } \\
                      &\le \sum_{o \in O \setminus A^*} c \ff{ A_{o} } / k \\
                      &\le \sum_{o \in O \setminus A^*} c\ff{ A_{m + 1} } / k \le c\ff{ A_{m + 1} }.
  \end{align*}
  From here, the result follows from Lemma \ref{lemm:aplusc}.
\end{proof}
Recall that $\altqs_c$ returns the set $A'$, the last $k$ elements added to $A$. 
The last portion of the proof shows that $f( A' )$ is a large
fraction of the value of $f(A_{m + 1})$; this part of the proof 
departs from the proof of Theorem \ref{thm:qs} above.
\begin{lemma} \label{lemm:Aprimec}
  Let $A'$ have its value after processing block $C_m$.
  Then
  $$\ff{ A_{m + 1} } \le \left( \frac{e}{e - (1 + c/k)^2} \right) \ff{A'}.$$.
\end{lemma}
\begin{proof}
  If $|A_{m + 1}| \le k$, $A' = A_{m + 1}$, and the lemma holds. Suppose $|A_{m + 1}| > k$. 
  Let $A' = \{a'_1, \ldots, a'_k \}$, in the order these elements
  were added to $A_{m + 1}$. Let $A'_i = \{ a_1', \ldots, a'_i \}$,
  $A'_0 = \emptyset$. Observe that by the condition on the marginal gain
  the addition of each block to $A$,
  $$f(A_{m + 1}) \ge (1 + c/k)^{\lfloor k / c \rfloor} \ff{ A_{m + 1} \setminus A' } \ge \frac{e}{( 1 + c/k )^{2}} \ff{ A_{m + 1} \setminus A' },$$
  by Claim \ref{claim:elementaryc}.
  Hence, by submodularity and nonnegativity of $f$, 
  \begin{equation} \label{ineq:3c}
    f( A' ) \ge f(A_{m + 1}) - \ff{ A_{m + 1} \setminus A' } \ge \left( \frac{e}{(1 + c/k)^2} - 1 \right) \ff{A_{m + 1} \setminus A' }.
  \end{equation}
  From (\ref{ineq:3c}), we have
  \begin{align*}
    \ff{A_{m + 1}} \le \ff{A_{m + 1} \setminus A'} + \ff{A'} &\le \left( \left( \frac{e}{(1 + c/k)^2} - 1 \right)^{-1} + 1 \right) \ff{A'} \\
    &= \left( \frac{e}{e - (1 + c/k)^2} \right) \ff{A' }. \qedhere
  \end{align*}
Since $k \ge 8c/e$,
Lemmas \ref{lemm:Agoodc} and \ref{lemm:Aprimec} show that the set $A'$ of
$\altqs_c$ maintains $f(A') \ge \altratio \opt_\uni$.
\end{proof}
\end{proof}

\section{Additional Empirical Evaluation} \label{apx:exp}

\subsection{Applications and Datasets}
The maximum cover objective is defined as follows. Suppose
$G = (V,E)$ is a graph. For any set $S \subseteq V$, 
let $S^I$ be the set of all vertices incident with any edge
incident with a vertex in $S$.
Then, define $$f(S) = \left| S^I \right| .$$
This objective is monotone and submodular.

The revenue maximization application uses the concave
graph model 
introduced in \citet{Hartline2008}. Given a social
network $G= (V,E)$ with nonnegative edge weights,
each user $u \in V$ is associated with a 
non-negative, concave function $f_u: \reals \to \reals$. 
In \citet{Hartline2008}, optimal marketing strategies
are defined, for which each user $u \in V$ has an associated
revenue function $R_u(S)$, which depends on the set $S$ of
players who have bought the item.
Thus, the total revenue from set $S$ is
$$f(S) = \sum_{u \in V} R_u( S ).$$
For this evaluation, we choose $R_u(S) = \left( \sum_{v \in S} w_{uv} \right)^{\alpha_u}$
where $\alpha_u$ is chosen independently for each $u$ uniformly in $(0,1)$.
The revenue maximization objective $f$ is monotone and submodular.

Network topologies are used from Stanford Large Network Dataset Collection \citep{snapnets}:
ca-Astro ($n = 18772$), a collaboration network
of Arxiv Astro Physics; ego-Facebook ($n=4039$); and as-Skitter ($n=1696415$).

\subsection{Additional Results}
\begin{figure}[t]
  \subfigure[]{ \label{fig:val-Facebookadd}
    \includegraphics[width=0.30\textwidth,height=0.15\textheight]{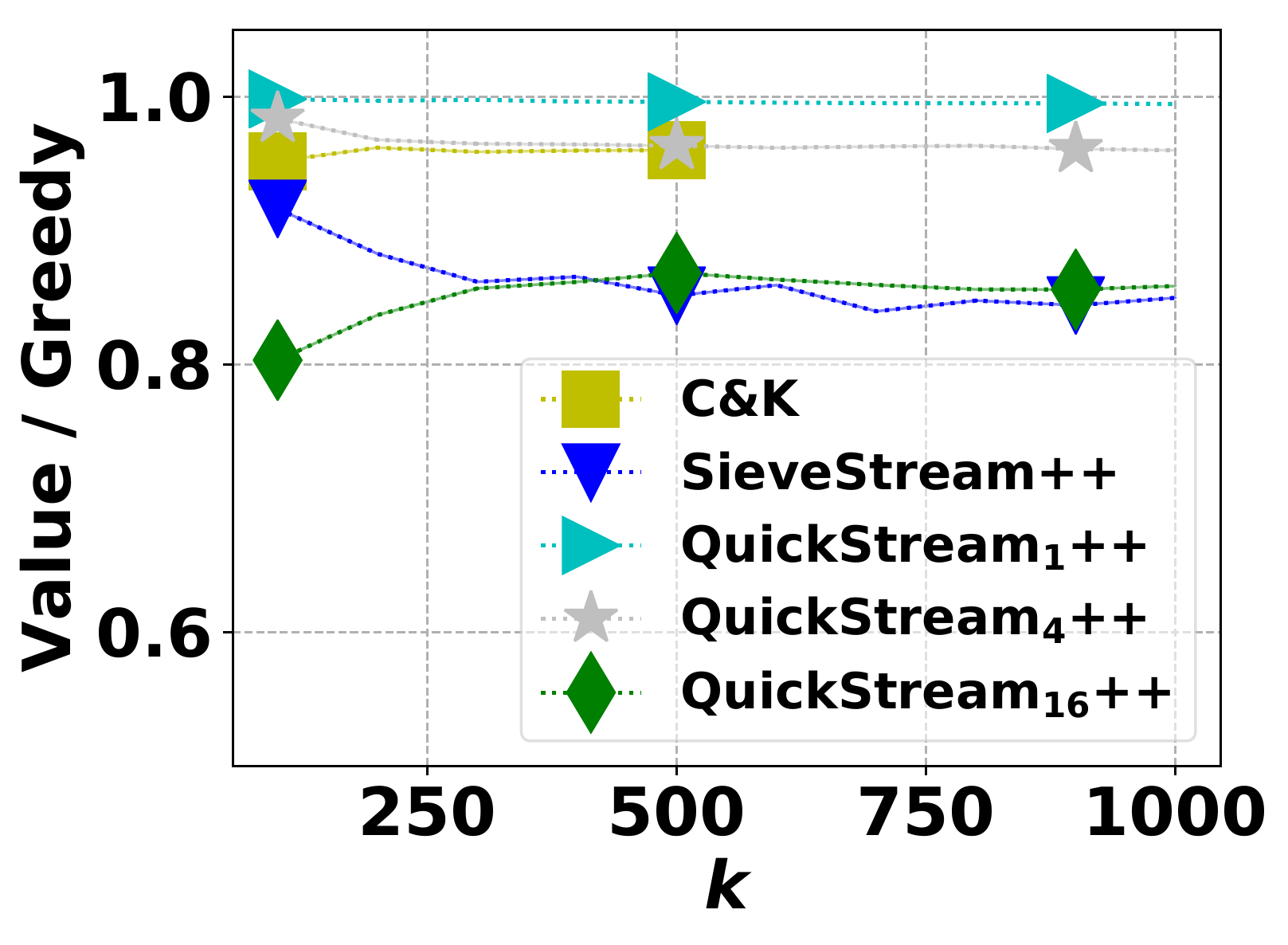}
   }
   \subfigure[]{ \label{fig:val-Facebookadd}
    \includegraphics[width=0.30\textwidth,height=0.15\textheight]{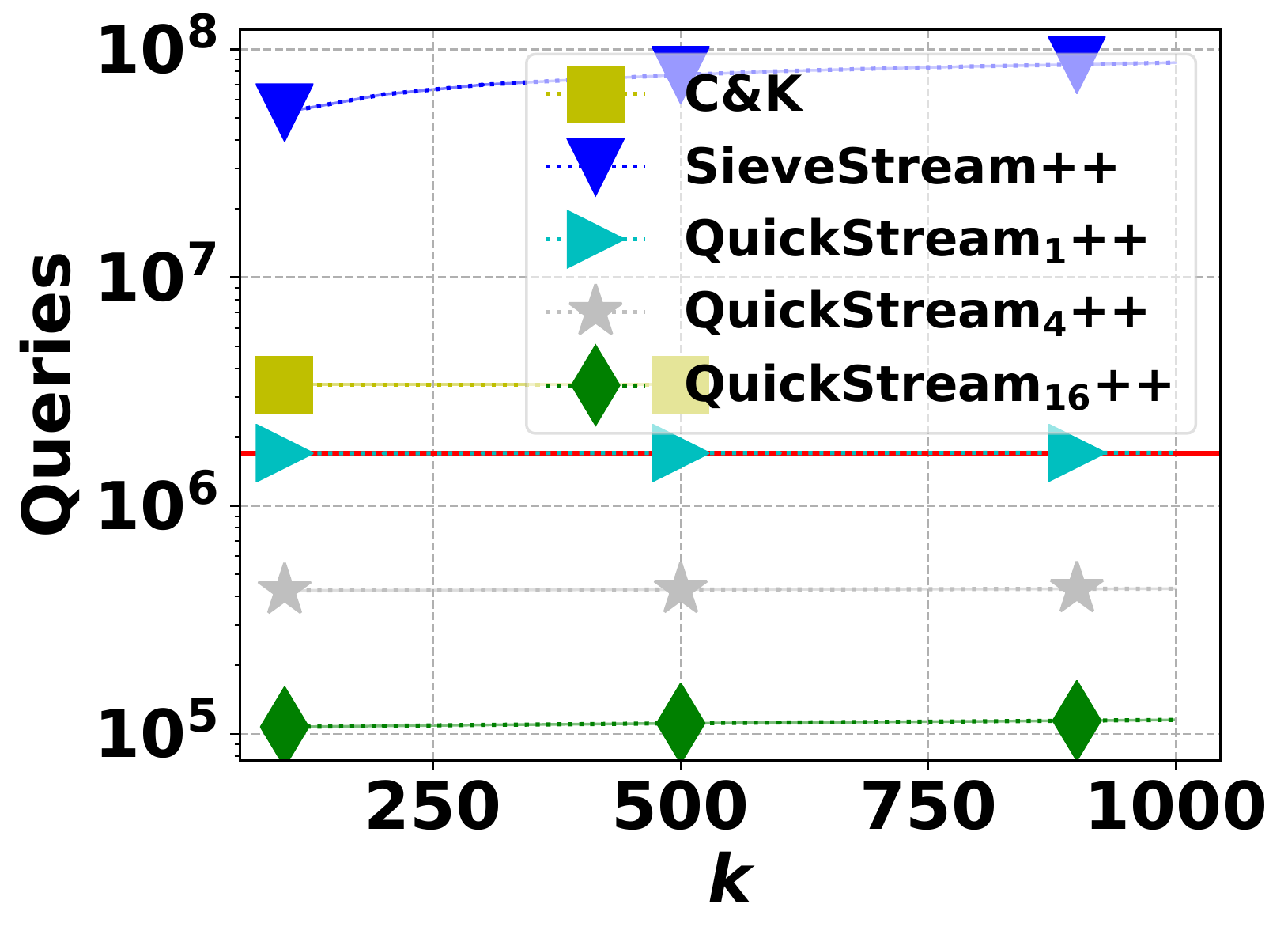}
   }
    \subfigure[]{ \label{fig:val-Facebookadd}
    \includegraphics[width=0.30\textwidth,height=0.15\textheight]{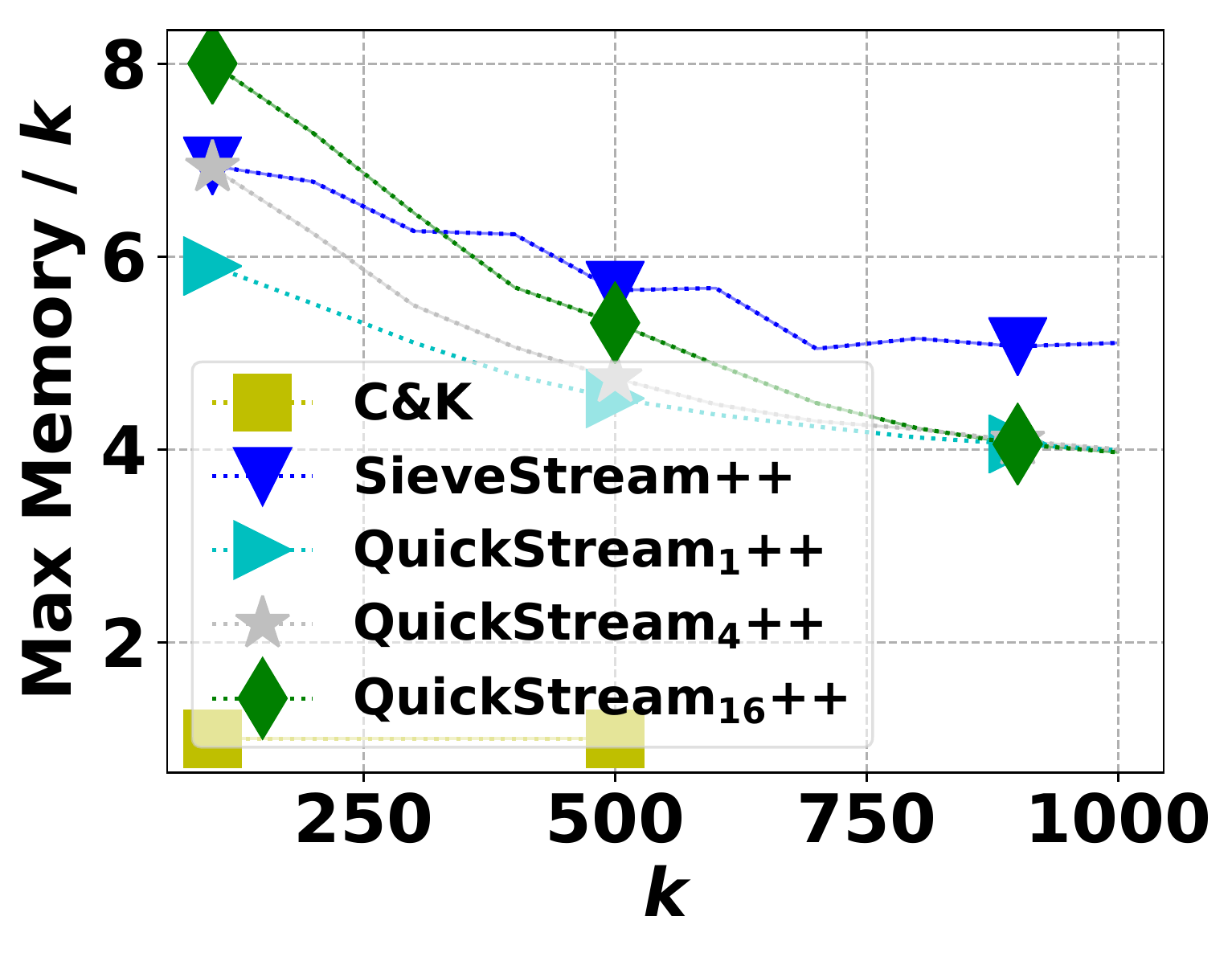}
   }

   \subfigure[]{ \label{fig:val-Facebookadd}
    \includegraphics[width=0.46\textwidth,height=0.20\textheight]{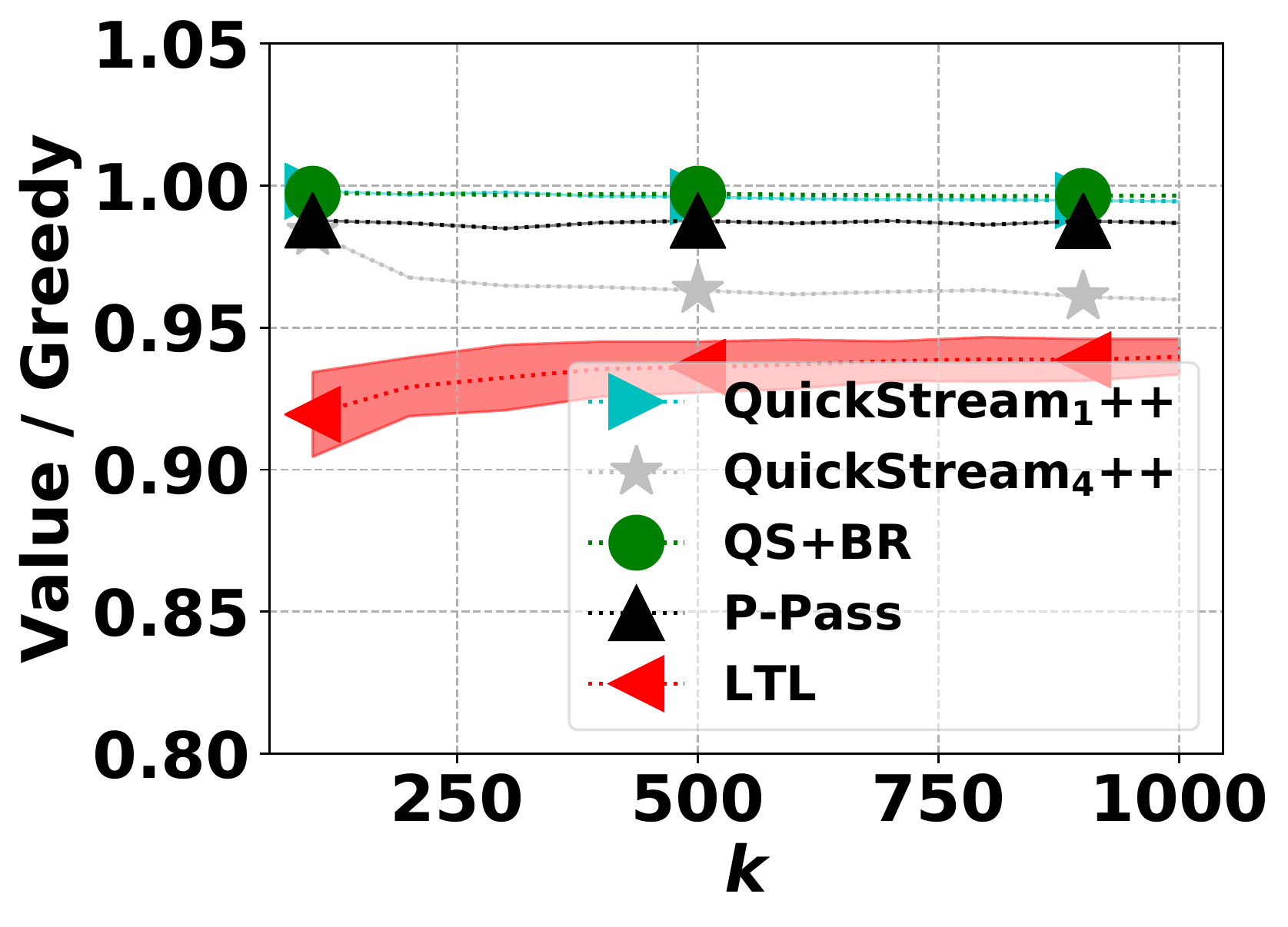}
   }
   \subfigure[]{ \label{fig:val-Facebookadd}
    \includegraphics[width=0.46\textwidth,height=0.20\textheight]{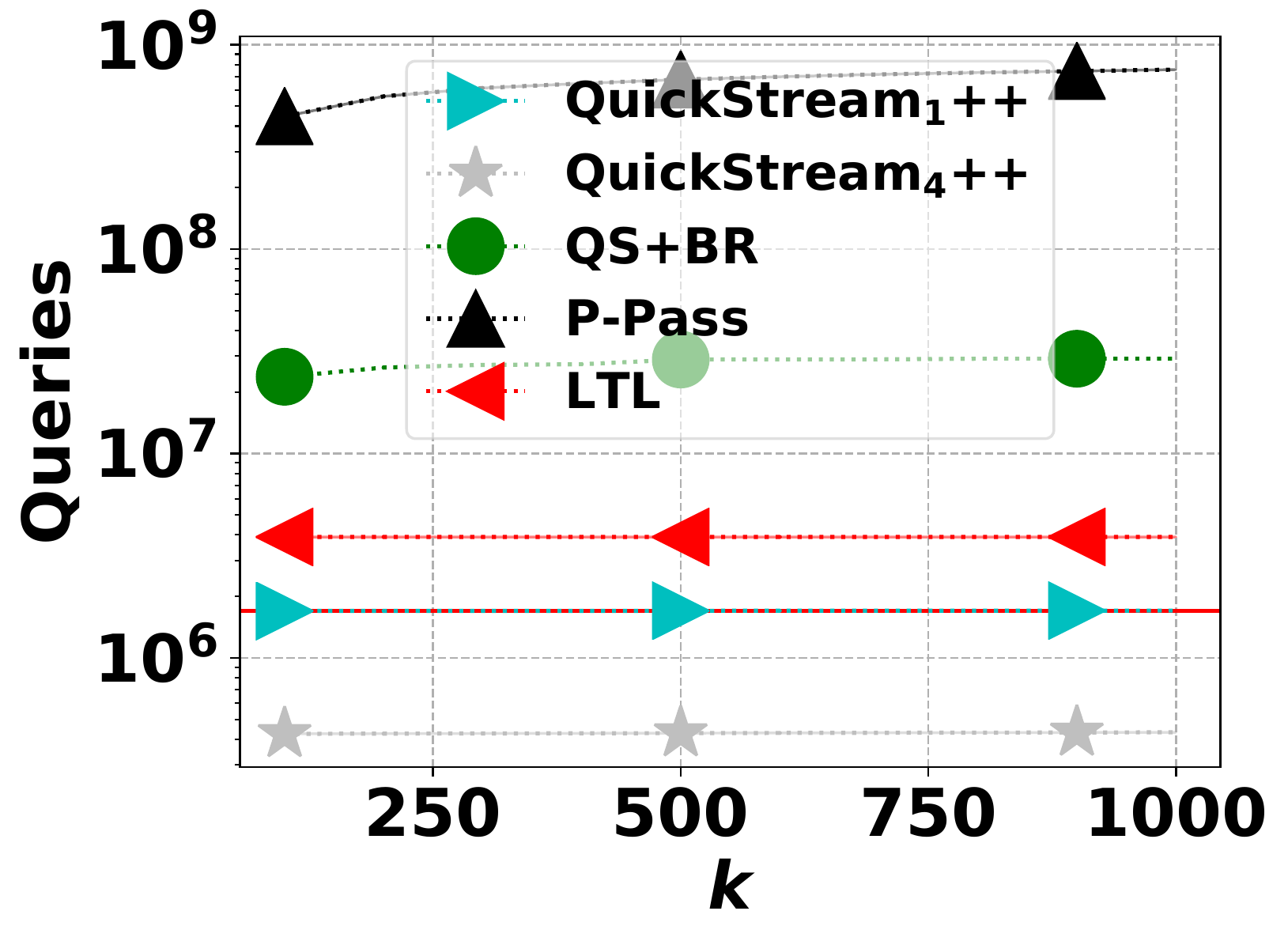}
   }


  \caption{Additional empirical results on the maxcover application on as-Skitter.} \label{fig:maxcov2}
\end{figure}
\begin{figure}[t]
  \subfigure[]{ \label{fig:val-Facebookadd}
    \includegraphics[width=0.30\textwidth,height=0.15\textheight]{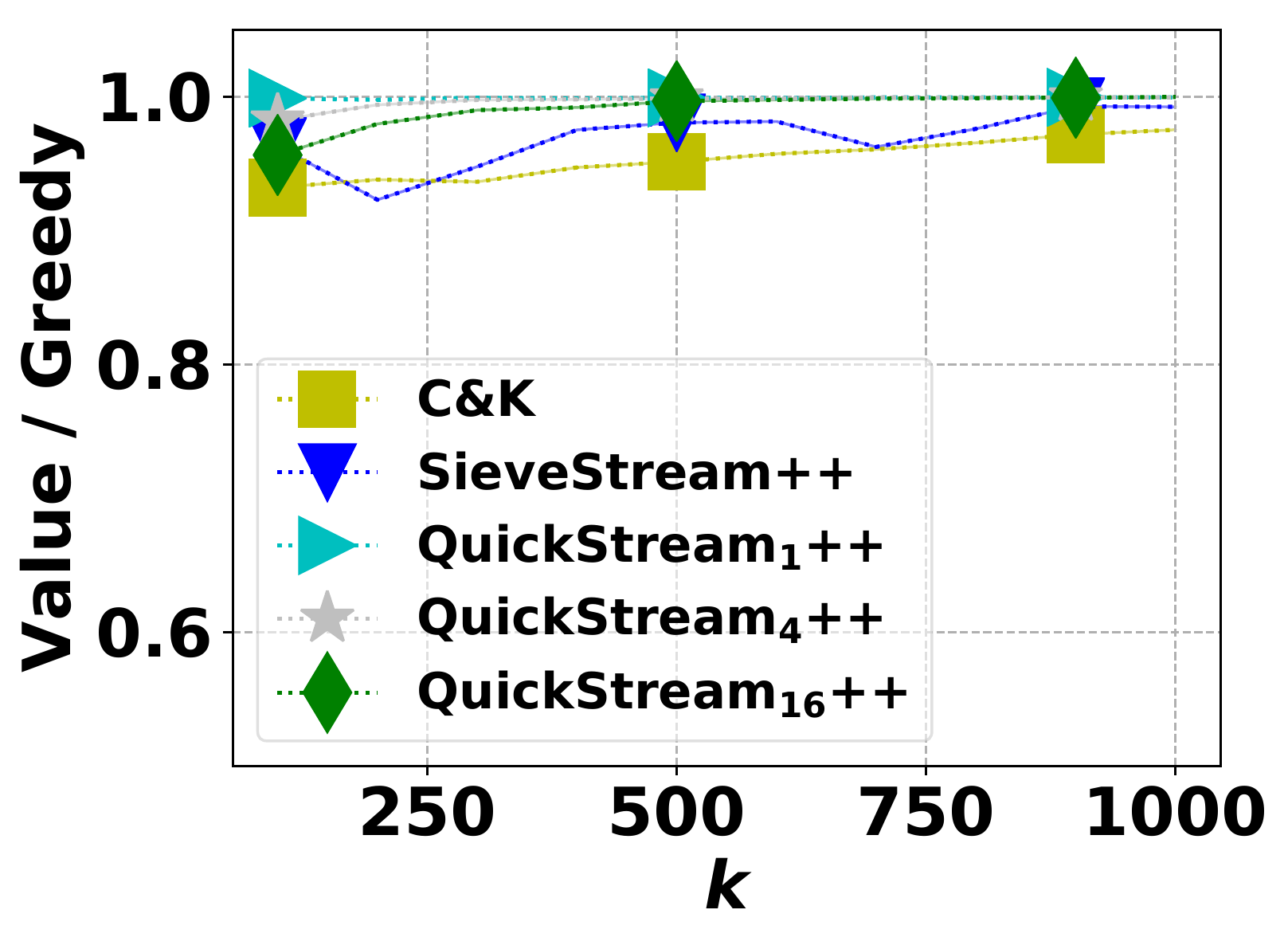}
   }
   \subfigure[]{ \label{fig:val-Facebookadd}
    \includegraphics[width=0.30\textwidth,height=0.15\textheight]{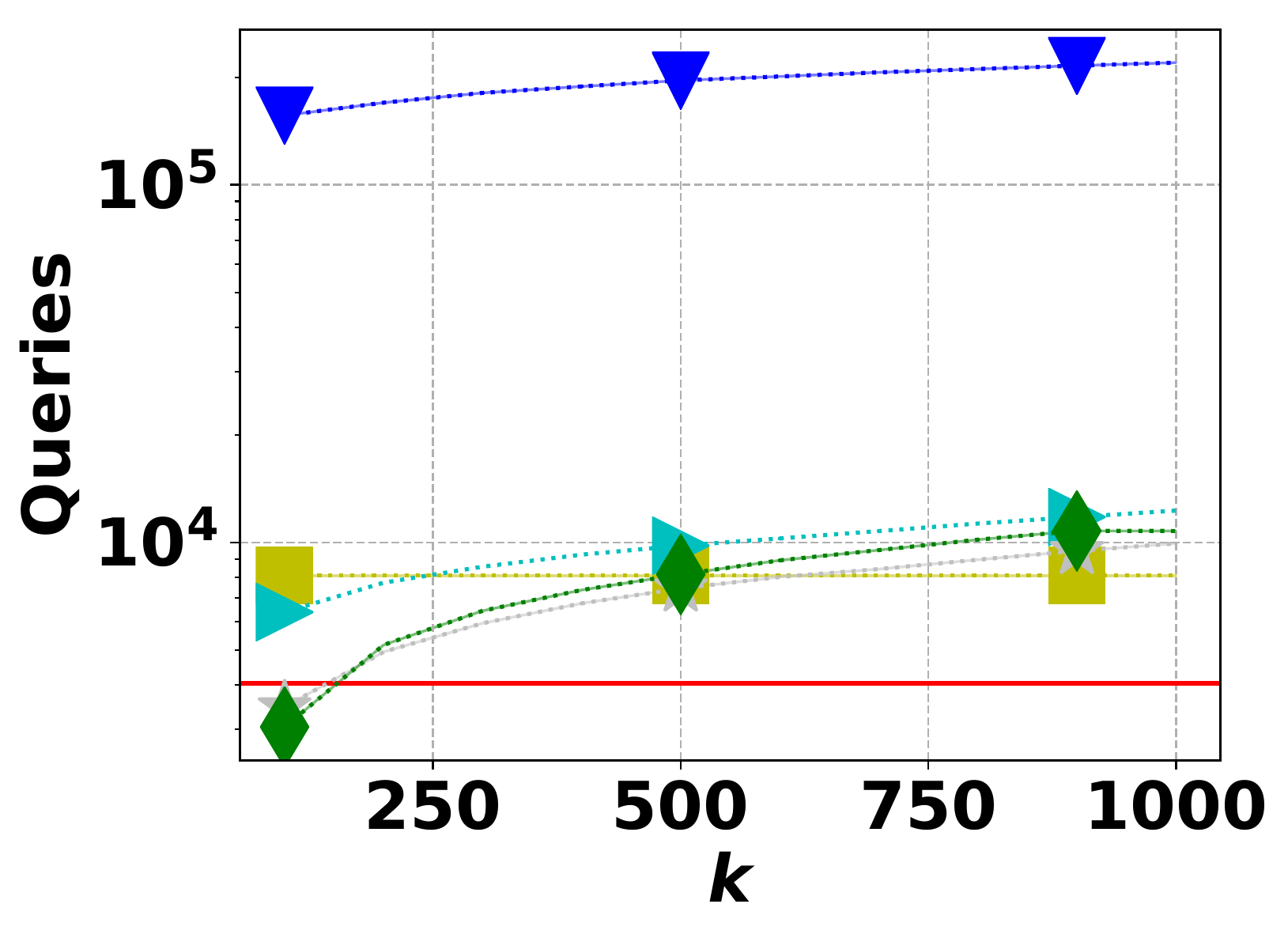}
   }
    \subfigure[]{ \label{fig:val-Facebookadd}
    \includegraphics[width=0.30\textwidth,height=0.15\textheight]{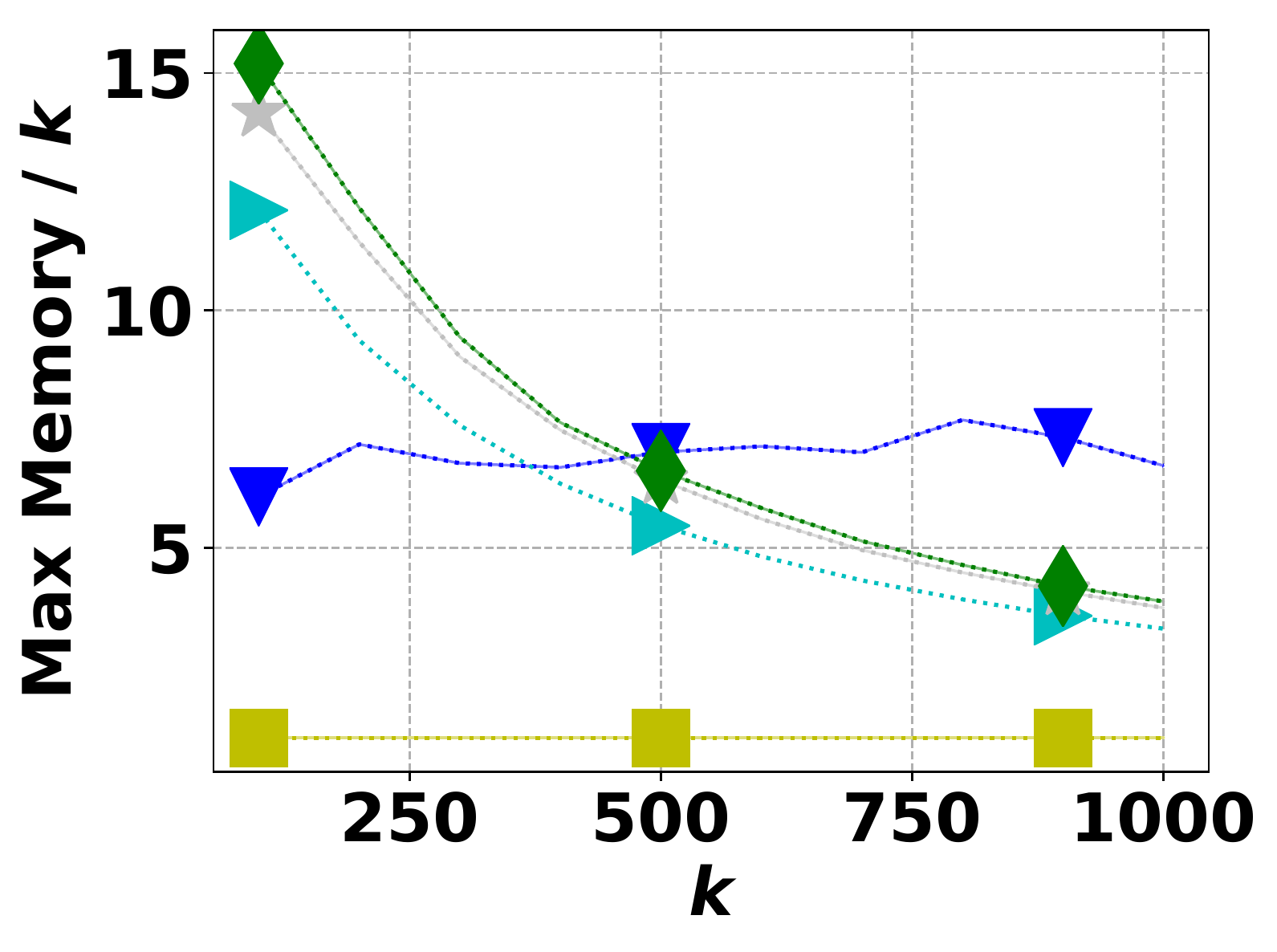}
   }

   \subfigure[]{ \label{fig:val-Facebookadd}
    \includegraphics[width=0.46\textwidth,height=0.20\textheight]{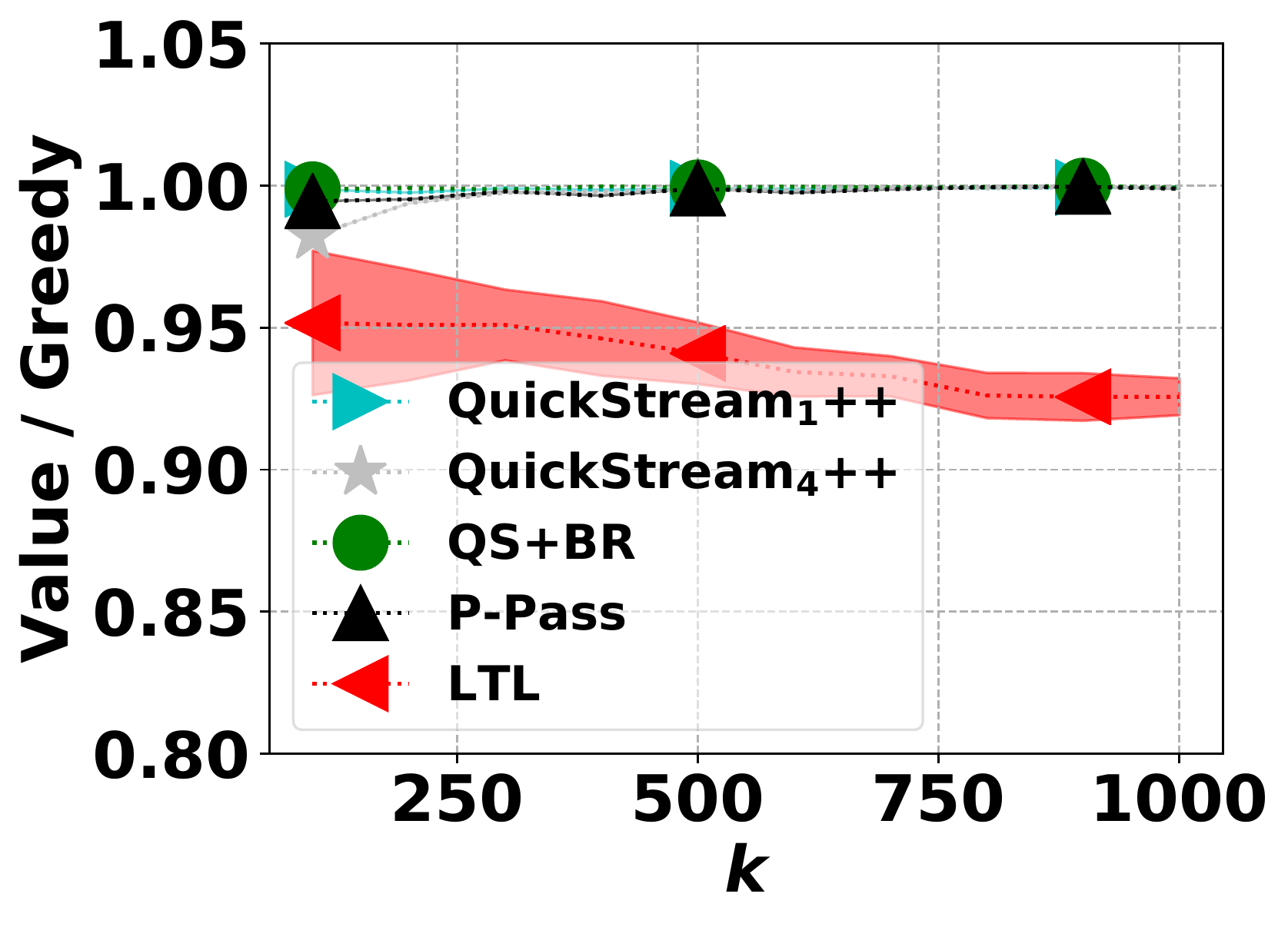}
   }
   \subfigure[]{ \label{fig:val-Facebookadd}
    \includegraphics[width=0.46\textwidth,height=0.20\textheight]{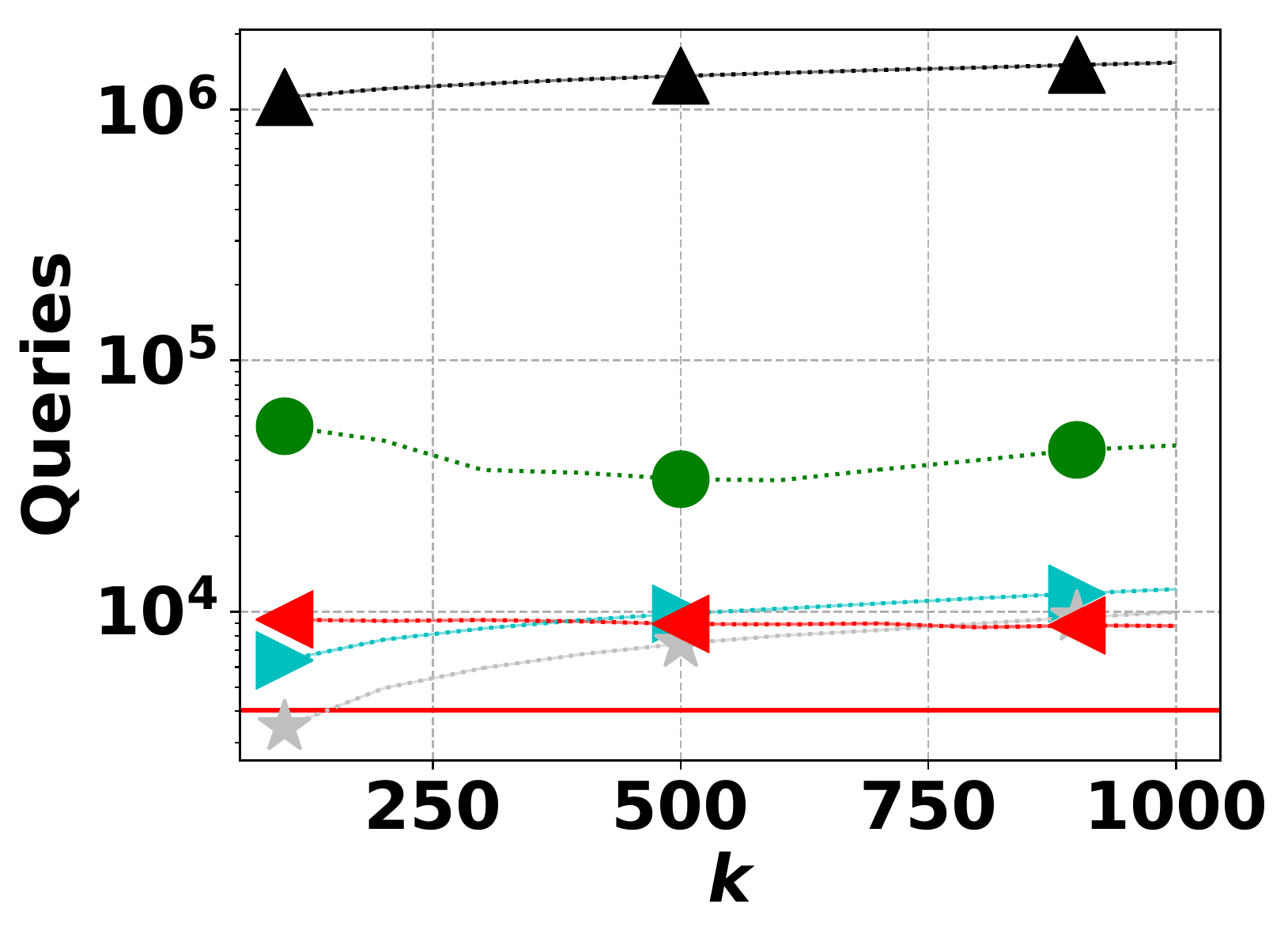}
   }


  \caption{Additional empirical results for the revenue maximization application on soc-Facebook.} \label{fig:revmax}
\end{figure}
\begin{figure}[t]
  \subfigure[]{ \label{fig:val-Astroadd}
    \includegraphics[width=0.30\textwidth,height=0.15\textheight]{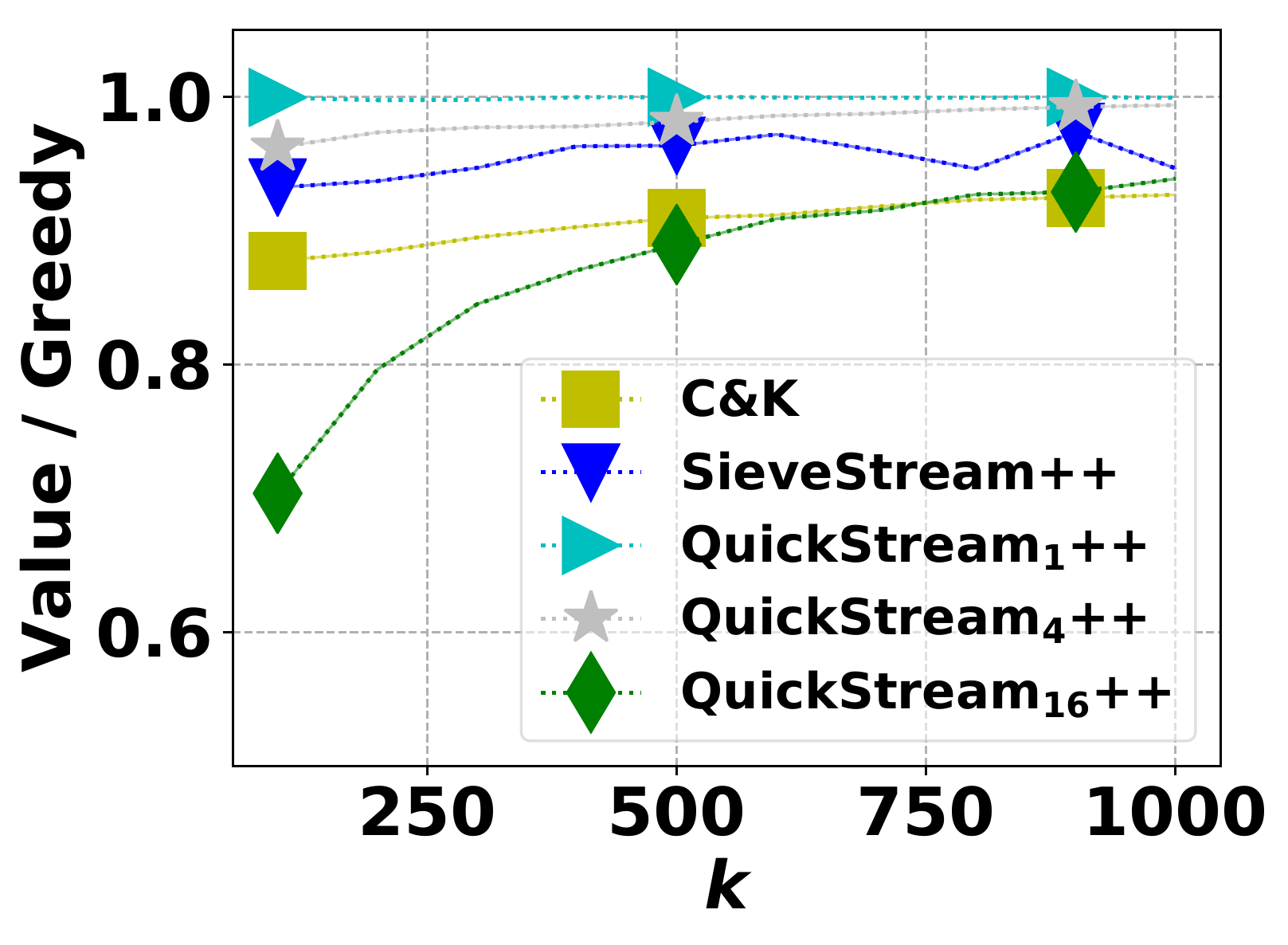}
   }
   \subfigure[]{ \label{fig:val-Astroadd}
    \includegraphics[width=0.30\textwidth,height=0.15\textheight]{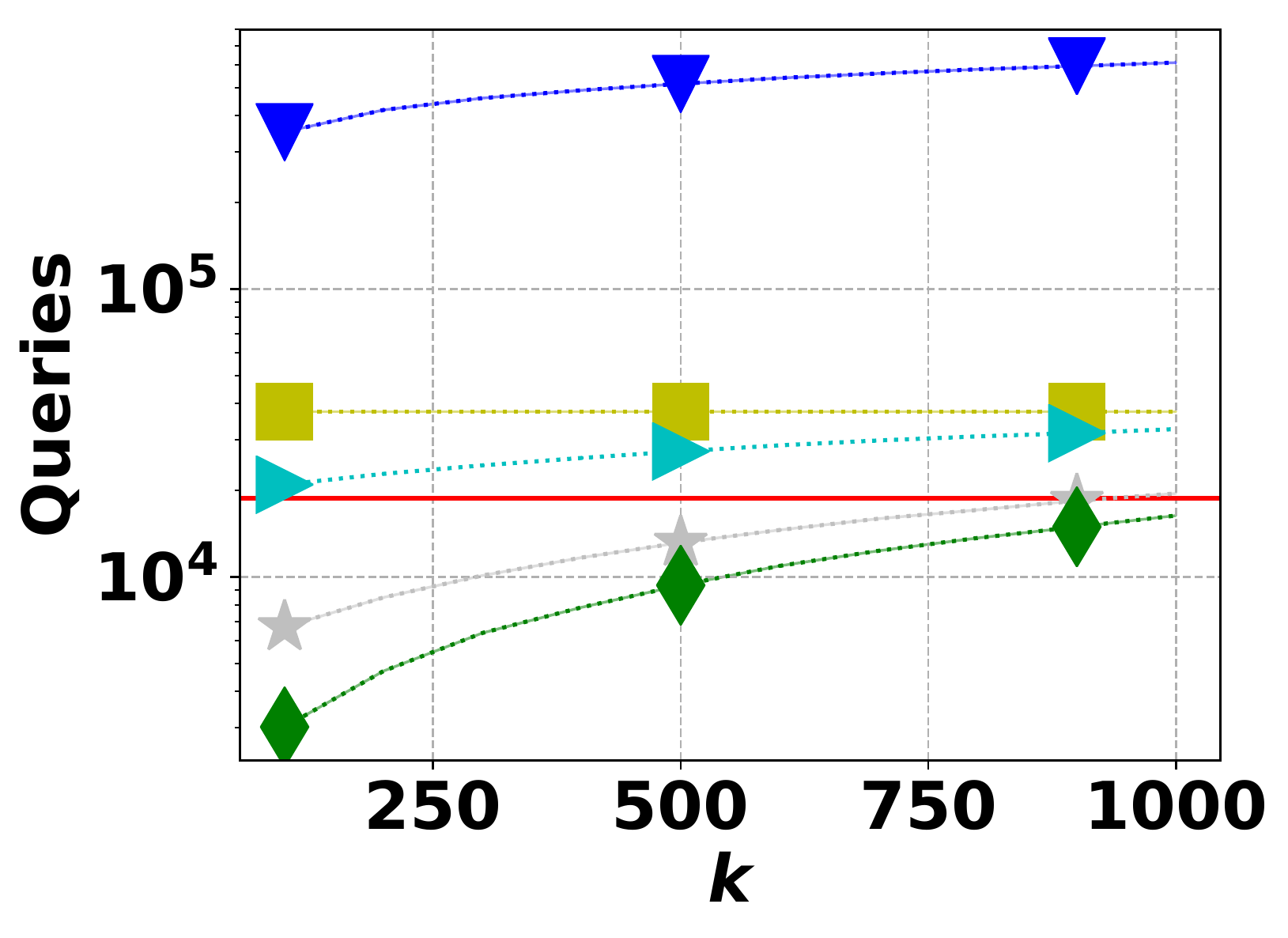}
   }
    \subfigure[]{ \label{fig:val-Astroadd}
    \includegraphics[width=0.30\textwidth,height=0.15\textheight]{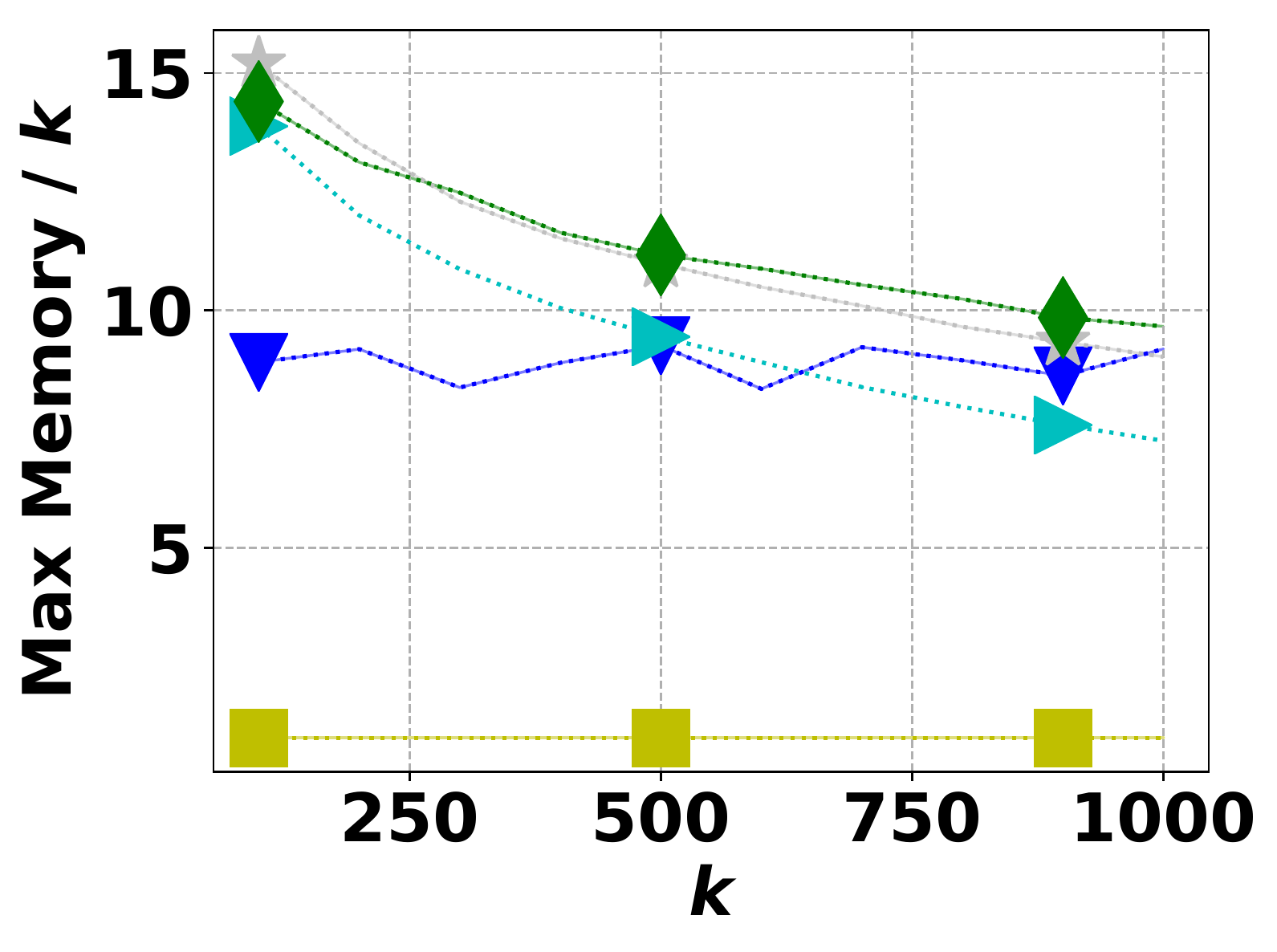}
   }

   \subfigure[]{ \label{fig:val-Astroadd}
    \includegraphics[width=0.46\textwidth,height=0.20\textheight]{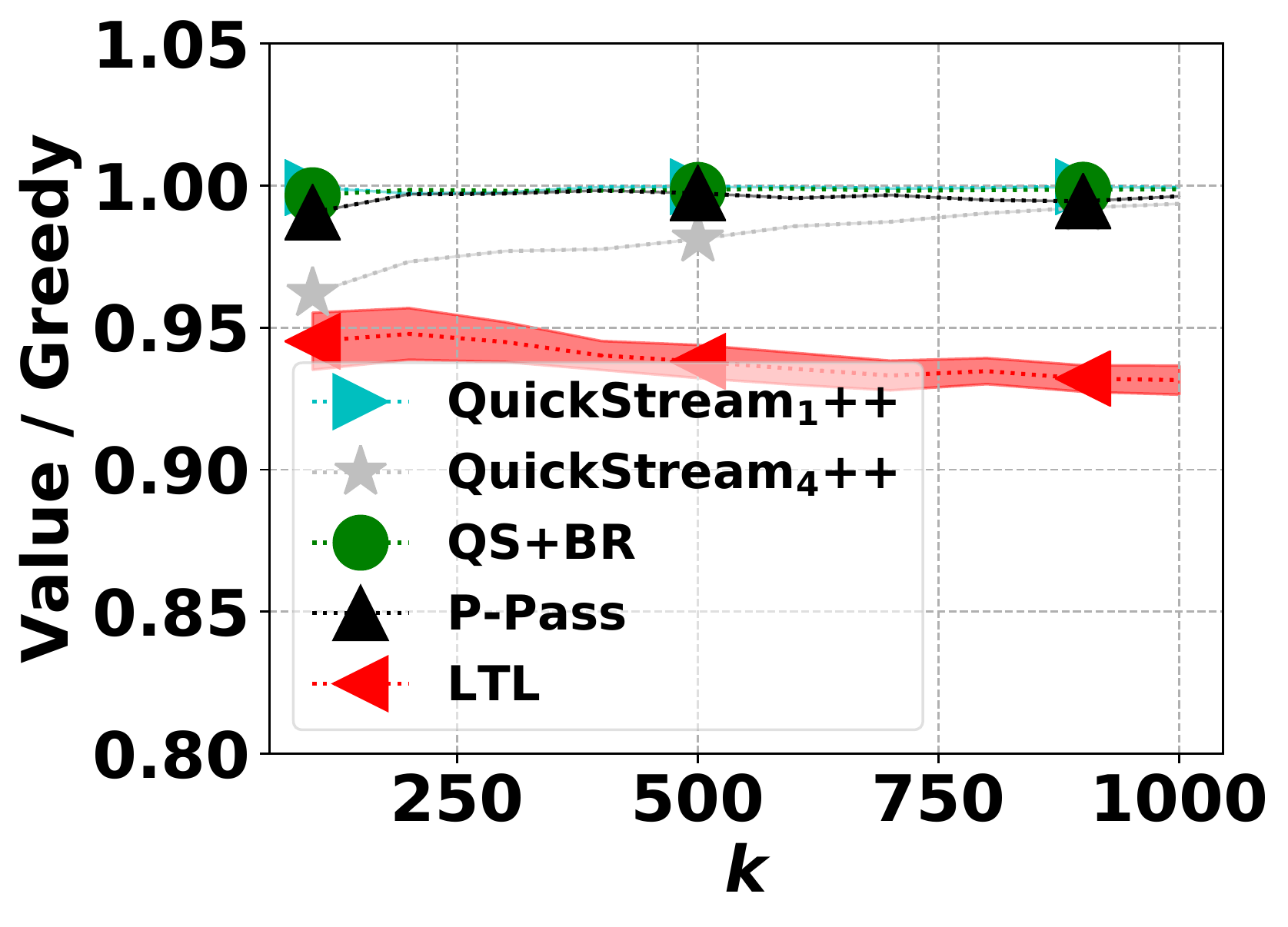}
   }
   \subfigure[]{ \label{fig:val-Astroadd}
    \includegraphics[width=0.46\textwidth,height=0.20\textheight]{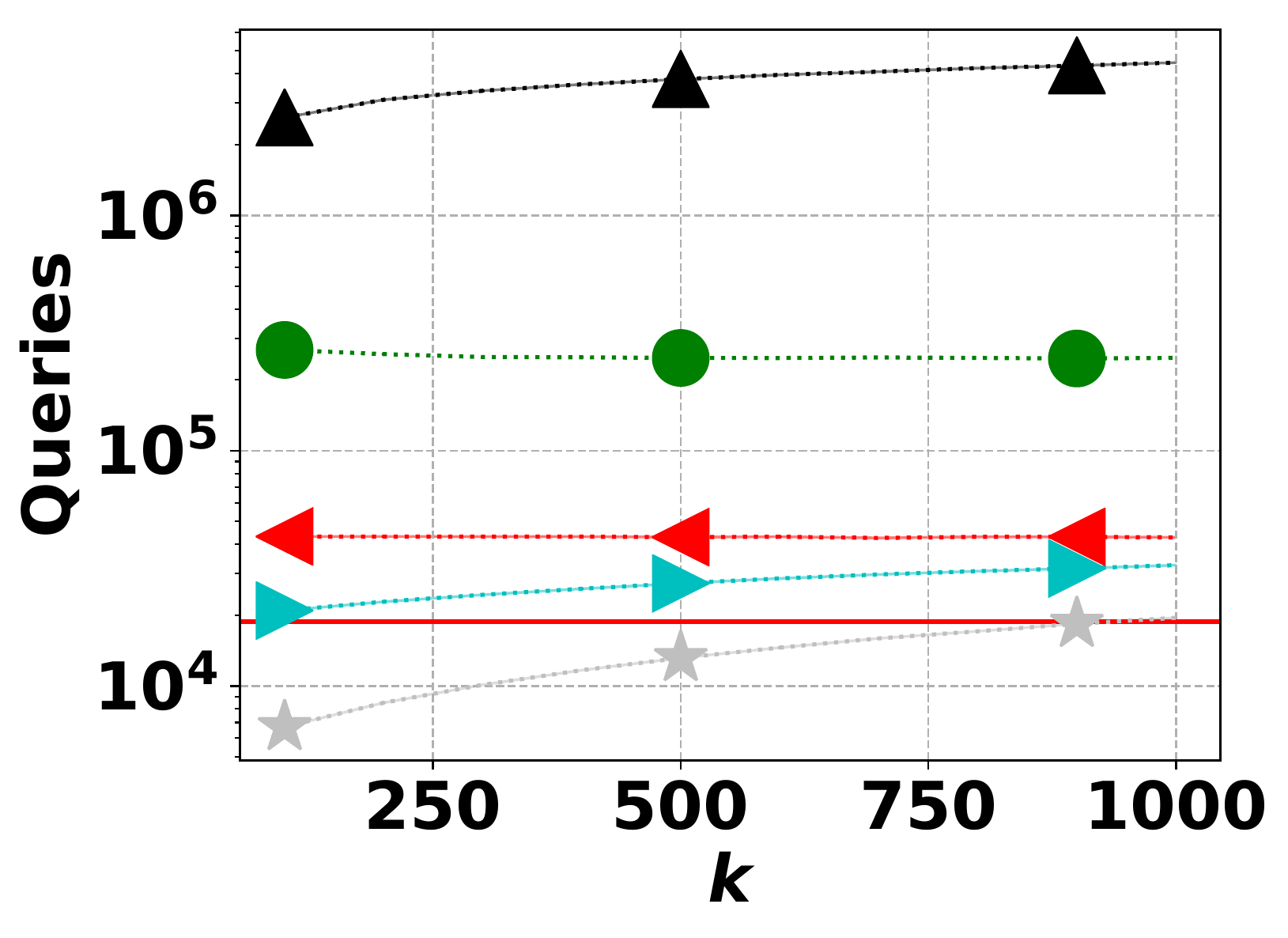}
   }


  \caption{Additional empirical results for the revenue maximization application on ca-AstroPh.} \label{fig:revmax2}
\end{figure}
Additional results from the maxcover application are shown in
Fig. \ref{fig:maxcov2}.
Results from the revenue maximization application are shown
in Figs. \ref{fig:revmax} and \ref{fig:revmax2}. These results are qualitatively similar
to the results from maximum coverage discussed in Section \ref{sec:exp}.

\end{document}